\documentclass{article}
\usepackage{arxiv}
\usepackage{hyperref}
\usepackage{microtype}

\usepackage{amssymb}
\usepackage{amsmath}
\usepackage{amsthm}
\usepackage{mathdots}
\usepackage{mathtools}

\usepackage{tensor}

\usepackage{urwchancal}
\DeclareFontFamily{OT1}{pzc}{}
\DeclareFontShape{OT1}{pzc}{m}{it}{<-> s * [1.10] pzcmi7t}{}
\DeclareMathAlphabet{\mathpzc}{OT1}{pzc}{m}{it}

\usepackage{graphicx}
\usepackage{ifpdf}
\ifpdf
\usepackage{epstopdf}
\PrependGraphicsExtensions{.svg}
\fi
\usepackage{xypic}

\newtheorem{theorem}{Theorem}[section]
\newtheorem{lemma}[theorem]{Lemma}

\newtheorem{proposition}[theorem]{Proposition}

\newtheorem{assumption}[theorem]{Assumption}

\newtheorem{definition}[theorem]{Definition}

\newtheorem{remark}[theorem]{Remark}

\providecommand{\N}{\mathbb{N}}
\providecommand{\R}{\mathbb{R}}

\providecommand{\SO}{\mathbf{SO}}
\providecommand{\SL}{\mathbf{SL}}

\providecommand{\SE}{\mathbf{SE}}

\providecommand{\grpG}{\mathbf{G}}

\providecommand{\gothg}{\mathfrak{g}}

\providecommand{\gothL}{\mathfrak{L}}
\providecommand{\gothR}{\mathfrak{R}}

\providecommand{\gothX}{\mathfrak{X}} % as in X(M)

\providecommand{\Sph}{\mathrm{S}}

\providecommand{\calB}{\mathcal{B}}

\providecommand{\calF}{\mathcal{F}}
\providecommand{\calG}{\mathcal{G}}

\providecommand{\calM}{\mathcal{M}}
\providecommand{\calN}{\mathcal{N}}

\providecommand{\calS}{\mathcal{S}}

\providecommand{\calV}{\mathcal{V}}

\providecommand{\vecV}{\mathbb{V}}

\providecommand{\calV}{\mathcal{V}}

\providecommand{\Id}{I} % identity of a matrix group.

\providecommand{\calf}{\mathpzc{f}} % extended input function - boldface f

\DeclareMathOperator{\spn}{span}

\DeclareMathOperator{\stab}{stab}
\DeclareMathOperator{\Ad}{Ad}

\DeclareMathOperator{\image}{im}

\providecommand{\id}{\mathrm{id}} % identity map
\providecommand{\tT}{\mathrm{T}} % tangent bundles
\providecommand{\td}{\mathrm{d}}
\providecommand{\tD}{\mathrm{D}}

\providecommand{\ddt}{\frac{\td}{\td t}}

\providecommand{\Fr}[2]{\left. \mathrm{D}_{#1} \right|_{#2}}

\providecommand{\mr}[1]{{#1}^\circ} % reference element.

\providecommand{\ob}[1]{\overline{#1}} % homogeneous vector
\usepackage{accents}
\makeatletter
\providecommand{\scirc}{%
    \hbox{\fontfamily{\rmdefault}\fontsize{0.4\dimexpr(\f@size pt)}{0}\selectfont{\raisebox{-0.52ex}[0ex][-0.52ex]{$\circ$}}}}
\makeatother
\mathchardef\mhyphen="2D
\providecommand{\etal}{\textit{et al.}~}

\newcommand{\publicationdetails}
{Submitted to SIAM Journal of Control and Optimisation}
\newcommand{\publicationversion}
{Preprint}
\begin{document}
%===============================================================================

\title{Equivariant Systems Theory and Observer Design}
\headertitle{Equivariant Systems Theory and Observer Design}
%% The headertitle will be printed in the header.  If the main title is too long define a shorter headertitle

\author{
\href{https://orcid.org/0000-0002-7803-2868}{\includegraphics[scale=0.06]{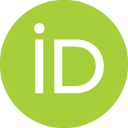}\hspace{1mm}
Robert Mahony}\\
	Department of Electrical, Energy and Materials Engineering\\
	Australian National University\\
    ACT, 2601, Australia \\
	\texttt{Robert.Mahony@anu.edu.au} \\
	\And	
\href{https://orcid.org/0000-0001-7383-5453}{\includegraphics[scale=0.06]{orcid.png}\hspace{1mm}
Tarek Hamel} \\
    $^2$ I3S-CNRS, \\
    University C\^ote d'Azur and Insitut Universitaire de France \\
    France\\
	\texttt{thamel@i3s.unice.fr} \\
	\And	
\href{https://orcid.org/0000-0002-5881-1063}{\includegraphics[scale=0.06]{orcid.png}\hspace{1mm}
    Jochen Trumpf} \\
	Department of Electrical, Energy and Materials Engineering\\
	Australian National University\\
    ACT, 2601, Australia \\
	\texttt{Jochen.Trumpf@anu.edu.au} \\
}
\maketitle

%===============================================================================
\begin{abstract}
A wide range of system models in modern robotics and avionics applications admit natural symmetries.
Such systems are termed equivariant and the structure provided by the symmetry is a powerful tool in the design of observers.
Significant progress has been made in the last ten years in the design of filters and observers for attitude and pose estimation, tracking of homographies, and velocity aided attitude estimation, by exploiting their inherent Lie-group state-space structure.
However, little work has been done for systems on homogeneous spaces, that is systems on manifolds on which a Lie-group acts rather than systems on the Lie-group itself.
Recent research in robotic vision has discovered symmetries and equivariant structure on homogeneous spaces for a host of problems including the key problems of visual odometry and visual simultaneous localisation and mapping.
These discoveries motivate a deeper look at the structure of equivariant systems on homogeneous spaces.
This paper provides a comprehensive development of the foundation theory required to undertake observer and filter design for such systems.
\end{abstract}


%===============================================================================
%%%%%%%%%%%%%%%%%%%%%%%%%%%%%%%%%%%%%%%%%%%%%%%%%%%%%%%%%%%%%%%%%%%%%%%%%%%%%%%
\section{Introduction}\label{sec:intro}

The design of global observers for mechanical systems is a core problem in the fields of control, robotics, and autonomous systems.
Such systems are almost always nonlinear and classical observer design methodologies in local coordinates, typically based on Extended Kalman Filters (EKF) \cite{anderson1979} or Unscented Kalman Filters (UKF) \cite{JulierUKF} suffer from robustness and performance limitations
\cite{2003_LaViolaJr_acc,2016_Barrau_arxive,2017_Barrau_tac}.
A key observation is that the nonlinearity of many systems of interest is due to the structure of the state space, often a Lie-group or homogeneous manifold, rather than complexity of the kinematics.
There is a rich history of work in which authors exploit the geometric structure of the state-space for observer and filter design for these systems.
Early work by Salcudean \cite{1991_salcudean_TAC} used the geometric structure of the quaternion group for  attitude estimation of a satellite.
Thienel \etal  \cite{thienel2003} added an analysis of observability and  bias estimation to the work of Salcudean \cite{1991_salcudean_TAC}.
Markley \cite{Markley2003} derived the so-called Multiplicative Extended Kalman Filter (MEKF) \cite{martin_salaun2010}.
Aghannan \etal \cite{2003_Aghannan_TAC,2004_Maithripala_acc} applied the tools from geometric mechanical systems and invariance.
These ideas were further developed by Bonnabel \etal \cite{2007_Bonnabel_cdc,2008_Bonnabel_TAC} for kinematic systems leading to the Invariant Extended Kalman Filter (IEKF) \cite{2017_Barrau_tac}.
In parallel Mahony \etal \cite{hamel2006,2008_Mahony_tac} developed the first non-linear observer for the rotation matrices with almost global asymptotic stability.
From this foundational work different observers have been derived exploiting the special orthogonal group $\SO(3)$ for attitude estimation \cite{2011_Madgwick,trumpf2012,2012_Grip,2014_Sanyal,2017_Berkane_Tayebi_Tac},  the special Euclidean group $\SE(3)$ for pose estimation \cite{2011_Hua_SE(3),2013_Bras,2015_Hua}, and the special linear group $\SL(3)$ for homography estimation \cite{2011_Hamel,2019_Hua,2020_Hua}.
A novel Lie-group for the velocity-aided attitude estimation problem was proposed in \cite{bonnabel2006} that allows second order kinematics in translation to be dealt with in the same structure as first order kinematics in attitude.
More recently, Lie-groups for the full second order kinematics on $\tT\SO(3)$ \cite{2019_Ng_cdc} and $\tT\SE(3)$ \cite{2020_Ng_ifac} have also been studied.

Leading a growing interest within the systems and control community in the Simultaneous Localisation and Mapping (SLAM) problem, Bonnabel \etal \cite{2016_Barrau_arxive} proposed a novel Lie group $\SE_n(m)$ to design an invariant Kalman Filter for the SLAM problem.
Parallel work by Mahony \etal \cite{2017_Mahony_cdc} proposed the same group structure along with a novel quotient manifold structure for the state-space of the SLAM problem.
Work by Zlotnik \etal \cite{2018_forbes_TAC} derived a geometrically motivated observer for the SLAM problem that includes estimation of bias in linear and angular velocity inputs.
Lourenco \etal \cite{2016_LouGueBatOliSil,2018_Lourenco_RAS} proposed an observer for the landmark points, or structure, separate from the robot pose, using landmark depth and bearing as separate components of the observer.
A new symmetry for SLAM was proposed in \cite{2019_vanGoor_cdc} that had the additional advantage that visual measurements are equivariant, and has led to almost globally stable observers \cite{2020_vanGoor_ifac} for the visual SLAM problem.
Of course, visual odometry and visual SLAM have a long history in the robotics community and come with a rich existing literature, see \cite{JonesSoato2011,2016_Cadena_TRO} and references therein.
However, the recent filters developed within the systems and control community that exploit the underlying symmetry for these problems show better consistency \cite{2016_Barrau_arxive,brossardInvariantKalmanFiltering2018} and larger basins of attraction
\cite{2016_LouGueBatOliSil,2018_Lourenco_RAS,2020_vanGoor_ifac} than classical filters.

The systems and observers discussed above are, apart from a few recent papers by the authors \cite{2017_Mahony_cdc,2020_Vangoor_cdc,2020_vanGoor_ifac,2019_vanGoor_cdc},
posed on matrix Lie-group state spaces.
Exploiting symmetry and invariance for observer design on smooth homogeneous spaces, that is manifolds that submit to a smooth group action, is much rarer in the literature.
An observer design on $\Sph^2$, an homogeneous space under rotation action by $\SO(3)$, was proposed in \cite{Metni2005}.
This observer was lifted up to $\SO(3)$ \cite{2008_Mahony_tac} to obtain the complementary filter for attitude that has been a mainstay of the aerial robotics industry.
The original invariant filter \cite{2008_Bonnabel_TAC} work was posed on homogeneous spaces although the examples considered at the time were posed on Lie-groups and recent developments \cite{2017_Barrau_tac} have been targeted at systems with matrix Lie group state space.
The underlying conceptual approach of defining a `lifted system' evolving on a symmetry group was formalised in \cite{RM_2013_Mahony_nolcos} but has not been exploited due to a lack of foundation theory and a paucity of real-world examples.
Recently, the work of Mahony \etal \cite{2017_Mahony_cdc} showed that the SLAM problem carries the structure of a homogeneous space, a perspective that underlies recent work by the authors
\cite{2020_Vangoor_cdc,2020_vanGoor_ifac,2019_vanGoor_cdc}.
Although these examples exist, there is no unified framework in the literature for studying observer design for systems on homogeneous spaces.

This paper contributes to the new field of \emph{equivariant systems theory} and develops theory targeted at observer design.
The paper proposes a design methodology \S\ref{sub:design_outline} that can be applied to any system for which the state space is a homogeneous space.
The potential application of the proposed methodology includes all examples discussed above as well as a vast collection of new problems in modern robotics and avionics that have yet to be defined.
The key contributions of the paper are:
\begin{itemize}
\item
A careful development of kinematics for physical systems leading to a global input-state system model derived from explicit velocity output structure.

\item
A formal development of equivariance of input-state systems along with an understanding of input actions and the concept of input extensions.

\item
A theory of equivariant systems, including the concept of equivariant system lift and proof that any equivariant system admits such a lift.

\item
A discussion of invariant errors for systems with homogeneous state that  motivates an observer/filter architecture with the observer state posed on the symmetry group.
\end{itemize}
We stop short of proposing a specific observer/filter design, halting the development at the point where the error dynamics are derived.
We stop here for several reasons, not the least of which is that the space limitations of the present venue had been reached.
It is also the case that there is no best observer design methodology for all systems, and proposing a specific design in this paper would detract from its  goal of providing a foundation theory for all observer design methodologies that exploit symmetry.
To save space we have also made the difficult decision not to include examples in this paper and rather refer the reader to our recent works
\cite{2017_Mahony_cdc,2019_Ng_cdc,2020_Vangoor_cdc,2020_vanGoor_ifac,2019_vanGoor_cdc,2020_Ng_ifac}.
These examples, especially those involving the visual SLAM problem, use the analysis framework and approach formalised in the present paper.

The paper has six sections including the introduction.
We begin with an overview of the design methodology proposed \S\ref{sub:design_outline}.
This provides both a detailed outline of the structure of the paper as well as a summary of the approach.
Following this, Section~\ref{sec:notation} introduces the notation.
Section~\ref{sec:kinematic_systems} provides a first principles discussion of what is a kinematic system and how to model velocity outputs to build an input-state system model for observer design.
A reader who is happy to accept a system $\dot{\xi} = f(\xi,v)$ with measured velocity $v$ inputs as the starting position for observer design could omit \S\ref{sec:kinematic_systems}.
However, we have found that a deep understanding of the modelling process, and in particular the role of velocity measurement in defining the input-state model, is enormously beneficial in practice.
Section~\ref{sec:equivariant_representation} introduces the concepts of symmetry and equivariance and does the heavy lifting for the paper in defining equivariant kinematics, extended inputs, system lifts and showing that an equivariant system lift always exists.
The final section \S\ref{sec:Observer_design} proposes an observer  architecture based on the lifted system, discusses invariant errors and derives the error dynamics.
We argue that the invariant error provides \emph{the} key tool in observer design for such systems.
We show that restricting to invariant errors leads to the natural formulation of the observer on the symmetry group (Theorem \ref{th:invariant_error}) and requires the concepts of equivariance and the input extension theory (Theorem \ref{th:equivariant_extension}) and equivariant lift (Theorem \ref{th:equivariant_lift_existance}) in order to simplify the error dynamics and provide tractable error-state kinematics for the observer design problem.

%~~~~~~~~~~~~~~~~~~~~~~~~~~~~~~~~~~~~~~~~~~~~~~~~~~~~~~~~~~~~~~~~~~~~~~~~~~~~~
\subsection{Observer design program}\label{sub:design_outline}

In this section we present a high level outline of the proposed observer design methodology that will initially act as  road-map for the results in the body of the paper and later as a summary of the approach.
To keep this section concise we use the notation and terminology developed later in the paper without explanation for the moment.

The problem considered is the design of a state observer for a continuous-time kinematic system evolving on a homogeneous space $\calM$.
The approach taken includes choosing and modelling sensors as part of the observer design problem.
As such, we do not start with an ordinary differential equation system model in the classical sense since such a model presupposes the structure of the velocity inputs and that in turn assumes availability of certain velocity measurements.
Instead, we begin with a (kinematic) behaviour $\calB$ on the time interval $[0,\infty)$ over a signal space $\tT \calM$, the tangent bundle of the smooth manifold $\calM$ (Def.~\ref{def:kinematic_system}) and then use the choice of velocity measurements to derive an input-state model.
The steps of the proposed design methodology are:
\begin{enumerate}
\item\label{item:analysis_state} Choose a sensor suite to provide a complete (Def.~\ref{def:complete_output_g}) velocity output $g : \tT\calM \to \vecV$ (Def.~\ref{def:velocity_output}) for an input vector space $\vecV$  and a configuration output $h : \calM \to \calN$ (Def.~\ref{def:config_output}) for an output manifold $\calN$.

\item\label{item:analysis_kinematics} Choose a linear system function $f : \vecV \to \calM$ (Def.~\ref{def:linear_system_funtion}) that is compatible with $g$ (Def.~\ref{def:systemfunction_compatible}) to obtain a kinematic model (Def.~\ref{def:kinematic_model}).

\item\label{item:analysis_symmetry} Choose a symmetry group $\grpG$ and transitive group action $\phi : \grpG \times \calM \rightarrow \calM$ (Def.~\ref{def:symmetry_group}).

\item\label{item:analysis_input_sym} Extend the input space via equivariance ($\calV \supset\vecV $ Def.~\ref{def:input_ext}) and define an input group action $\psi : \grpG \times \calV \rightarrow \calV$ and extended equivariant system function $\calf : \calV \to \gothX(\calM)$ (Th.~\ref{th:equivariant_extension}).

\item\label{item:analysis_lift} Construct an equivariant lift function $\Lambda : \calM \times \calV \rightarrow \gothg$ (Def.~\ref{def:equivarian_lift}) such that $\td \phi_\xi \Lambda(\xi,v) = \calf(\xi,v)$ and $\Ad_X \Lambda(\phi_X(\xi),\psi_X(v)) =  \Lambda(\xi,v)$ (Th.~\ref{th:equivariant_lift_existance}).

\item\label{item:observer_architecture} Choose an origin $\mr{\xi} \in \calM$ and define the observer system  (Def.~\ref{def:observer_architecture})
\begin{align}
\dot{\hat{X}} & = \td L_{\hat{X}} \Lambda(\phi_{\hat{X}}(\mr{\xi}),v) + \td R_{\hat{X}} \Delta_t(\hat{X},y) \label{eq:plan_obs_system}\\
\hat{\xi} & = \phi_{\hat{X}} (\mr{\xi}) \label{eq:plan_state_est}
\end{align}
for $y = h(\xi)$ and an innovation map $\Delta_t : \grpG \times \calN \to \gothg$ that remains to be designed.

\item\label{item:error_dynamics}
Define the intrinsic error $e := \phi(\hat{X}^{-1},\xi)$ (Def.~\ref{def:error}) and compute the error dynamics \eqref{eq:equivariant_error_dynamics}
\begin{align}
\dot{e} = \td \phi_e \left( \Lambda(e,\psi_{\hat{X}^{-1}}(v)) - \Lambda(\mr{\xi},\psi_{\hat{X}^{-1}}(v)) \right)
- \td \phi_e \Delta_t( \hat{X},h(\phi_{\hat{X}}(e))
\label{eq:plan_error_dyn}
\end{align}

\item\label{item:observer_design}
Design an innovation $\Delta_t$ to make \eqref{eq:plan_error_dyn} converge to the origin $e \to \mr{\xi}$.
\end{enumerate}

The innovation designed in Step \ref{item:observer_design} is implemented in \eqref{eq:plan_obs_system} and the state estimate generated is  \eqref{eq:plan_state_est}.

Solutions to Steps \ref{item:analysis_state} and \ref{item:analysis_kinematics} are developed in Section \ref{sec:kinematic_systems}.
Solutions to Steps \ref{item:analysis_symmetry}, \ref{item:analysis_input_sym} and \ref{item:analysis_lift} are covered in Section \ref{sec:equivariant_representation}.
The observer framework, Steps \ref{item:observer_architecture} and \ref{item:error_dynamics}, is developed in Section \ref{sec:Observer_design}.
Step \ref{item:observer_design} is discussed in Section \ref{sec:Observer_design}, although as mentioned  earlier, this paper does not propose a particular choice for the innovation function.

%%%%%%%%%%%%%%%%%%%%%%%%%%%%%%%%%%%%%%%%%%%%%%%%%%%%%%%%%%%%%%%%%%%%%%%%%%%%%%%
\section{Notation and Preliminaries}\label{sec:notation}

Let $\calM$ and $\calN$ be smooth manifolds.
The tangent space at a point $\xi \in \calM$ is denoted $\tT_\xi \calM$.
A tangent vector $\eta_\xi \in \tT_\xi \calM$ is written with subscript identifying which specific tangent space $\tT_\xi \calM$ it lies in.
The tangent bundle is the collection $\tT \calM = \{\eta_\xi \;|\; \eta_\xi \in \tT_\xi \calM, \xi \in \calM\}$ equipped with the induced differential structure.
The bundle map $\pi : \tT \calM \rightarrow \calM$, $\pi (\eta_\xi) := \xi$, projects tangent vectors down to their base point.

For a smooth map $h : \calM \rightarrow \calN$ we denote the Fr\'{e}chet derivative with respect to a variable $\zeta \in \calM$, evaluated at $\xi \in \calM$ in direction $\eta_\xi \in \tT_\xi \calM$ by
$\Fr{\zeta}{\xi} h(\zeta)[\eta_\xi]$.
The differential of a smooth map maps between tangent bundles, is denoted $\td h : \tT\calM \rightarrow T\calN$, and is evaluated pointwise by
\[
\td h \eta_\xi = \Fr{\zeta}{\xi}h(\zeta)[\eta_\xi].
\]

A smooth vector field on $\calM$ is a smooth map $f : \calM \rightarrow \tT \calM$ such that $f(\xi) \in \tT_\xi \calM$.
The set of all smooth vector fields is denoted $\gothX(\calM)$ and is a linear (infinite dimensional) vector space over the field $\R$ under pointwise addition and scalar multiplication.
The Lie-bracket of two vector fields $f, g \in \gothX(\calM)$ is an anti-symmetric bracket $[f,g] \in \gothX(\calM)$ defined by $[f,g](x) = \Fr{z}{x}f(z) [g(x)] - \Fr{z}{x}g(z) [f(x)]$.
This bracket satisfies the Jacobi identity and $\gothX(\calM)$ is a Lie algebra \cite{1983_Warner}.

Let $\grpG$ be a finite-dimensional real Lie group; that is, a smooth manifold
endowed with a group multiplication and inverse operation which are smooth in the differential structure of the manifold.
For arbitrary $A, B \in \grpG$, the group multiplication is denoted by $AB$,
the group inverse by $A^{-1}$, and $\Id$ denotes the identity element of $\grpG$.
Define the left translation on the group by $L_A \colon \grpG \rightarrow \grpG$, $L_A B := AB$.
The right translation $R_A B := BA$ is analogous.
The differential of left (analog.~right) translation $\td L_A : \tT_B \grpG \to \tT_{AB} \grpG$ maps between tangent spaces.

For any vector $V \in \tT_{\Id} \grpG$ associate a vector field
$\td L_X V \in \gothX(\grpG)$ termed a \emph{left invariant} vector field.
The set of all such vector fields are denoted $\gothL(\grpG)$ and form a finite dimensional (due to the one-to-one correspondence with $\tT_\Id \grpG$) subspace $\gothL(\grpG) \subset \gothX(\grpG)$.
The left invariant vector fields $\gothL(\grpG)$ are closed under the Lie-bracket on $\gothX(\grpG)$ and form a finite-dimensional Lie-algebra.
Identifying each left-invariant vector field with an element of the identity tangent space induces a Lie-bracket on $\tT_\Id \grpG$ via
\begin{align}
[U,V] := \td L_{X^{-1}} [\td L_X U, \td L_X V]
\label{eq:Lie_bracket}
\end{align}
where the right hand side is the Lie-bracket of vector fields in $\gothX(\grpG)$ and the left hand side is an anti-symmetric operator on $\tT_\Id \grpG$ \cite{1983_Warner}.
In this paper, we will use the notation $\gothg$ to denote $\tT_\Id \grpG$ along with the left induced bracket \eqref{eq:Lie_bracket} and refer to this as the Lie-algebra of $\grpG$.

The set of right invariant vector fields $\{ \td R_X U \: |\; U \in \tT_\Id \grpG \}$ also forms a Lie sub-algebra $\gothR(\grpG) \subset \gothX(\calM)$.
The set of right and left invariant vector fields are not equal unless $\grpG$ is Abelian.
In this paper we distinguish strongly between $\gothL(\grpG)$, $\gothR(\grpG) \subset \gothX(\grpG)$ and $\gothg \equiv \tT_\Id \grpG$ along with the left induced bracket.

The family of smooth maps $I_X \colon \grpG \rightarrow \grpG$, for $X \in \grpG$,
\begin{align}
I_X(Z) := X Z X^{-1} = L_X R_{X^-1} Z, \quad  Z \in \grpG
  \label{eq:I_X}
\end{align}
are termed \emph{inner automorphisms} of $\grpG$.
The family of linear maps $\Ad_X \colon \gothg \rightarrow \gothg$, for $X \in \grpG$,
\begin{align}
\Ad_X (U) := \Fr{Z}{\Id} I_X(Z)[U],\quad U \in \gothg,
\label{eq:Ad_X}
\end{align}
are termed the \emph{Adjoint maps} (written with upper case `A') of $\gothg$.

A right group action $\phi$ of $\grpG$ on a smooth manifold $\calM$ is a smooth mapping
\begin{align*}
\phi & \colon \grpG \times \calM \rightarrow \calM,
\end{align*}
with $\phi(A,\phi(B,x)) = \phi(BA,x)$ and $\phi(I,x) = x$.
A left group action is analogous with $\phi(A,\phi(B,x)) = \phi(AB,x)$.
A group action induces families of smooth diffeomorphisms $\phi_A \colon \calM \rightarrow \calM$ for $A \in \grpG$ by $\phi_A(x) := \phi(A,x)$, and smooth (nonlinear) projections $\phi_x \colon \grpG \rightarrow \calM$ for $x \in \calM$ by $\phi_x(A) := \phi(A,x)$.
The group action $\phi$ is termed \emph{transitive} if $\phi_x$ is surjective and in this case the manifold $\calM$ is termed a \emph{homogeneous space} of $\grpG$ \cite{Boo86}.
The group action is termed \emph{effective} if the only element $A \in \grpG$ such $\phi(A,\xi) = \xi$ for all $\xi \in \calM$ is the identity $A = \Id$ \cite{Boo86}.
The group action is termed \emph{free} if for all $\xi \in \grpG$ the only element $A \in \grpG$ such $\phi(A,\xi) = \xi$ is the identity $A = \Id$.
For concatenation of group actions $\phi$ we write $\phi_X \phi_Z (v) = \phi_X(\phi_Z(\xi))$ to simplify notation.
For a group action $\phi \colon \grpG \times \calM \rightarrow \calM$, the stabilizer of an element $x \in \calM$ is given by
$\stab_\phi (x) = \{ A \in \grpG\;|\; \phi(A,x)=x\}$, and is a subgroup of $\grpG$.
Let $X \in \grpG$, then for a right action
$\stab_\phi (\phi_X(\mr{\xi})) = I_{X^{-1}} \stab_\phi (\mr{\xi})$.

\begin{proposition}\label{prop:homog_vel_transf}
Let $\calM$ be a homogeneous space with respect to a Lie group $\grpG$ and group action $\phi : \grpG\times \calM\rightarrow \calM$.
Then the following diagram commutes \\
\[
\xymatrix{
	\gothg \ar@{->}[r]^{\Ad_{X^{-1}}} \ar@{->}[d]_{\td \phi_{\mr{\xi}}} &
	\gothg \ar@{->}[d]^{\td \phi_{\phi_X (\mr{\xi})}} \\
	\tT_{\mr{\xi}} \calM \ar@{->}[r]^{\td \phi_{X}} &
	\tT_{\phi_X (\mr{\xi})} \calM
	}
\]
In particular, for $\mr{\xi} \in \calM$, $X \in \grpG$ and $U \in \gothg$ then
\begin{align}
\td \phi_X\td \phi_{\mr{\xi}} U & = \td \phi_{\phi_X(\mr{\xi})} \Ad_{X^{-1}} U \label{eq:homog_commutative}
\end{align}
\end{proposition}

\begin{proof}
Direct computation yields
\begin{align*}
\td \phi_X \td \phi_{\mr{\xi}} U & = \td \phi_X  \Fr{Z}{\Id} \phi (Z,\mr{\xi}) [U] \\
& = \Fr{Z}{\Id} \phi(X, \phi (Z,\mr{\xi})) [U] \\
& = \Fr{Z}{\Id}  \phi (Z X,\mr{\xi})  [U] \\
& = \Fr{Z}{\Id}  \phi (X X^{-1} Z X ,\mr{\xi}) [U] \\
& = \Fr{Z}{\Id}  \phi(X^{-1} Z X, \phi (X  ,\mr{\xi})) [U] \\
& = \Fr{Z}{\Id}  \phi_{\phi_X  (\mr{\xi})}(X^{-1} Z X) [U] \\
& = \td \phi_{\phi_{X}(\mr{\xi})}  \Fr{Z}{\Id}  (X^{-1} Z X) [U] \\
& = \td \phi_{\phi_{X}(\mr{\xi})} \Ad_{X^{-1}} U
\end{align*}
which completes the proof.
\end{proof}

%%%%%%%%%%%%%%%%%%%%%%%%%%%%%%%%%%%%%%%%%%%%%%%%%%%%%%%%%%%%%%%%%%%%%%%%%%%%%%
\section{Kinematic Systems}\label{sec:kinematic_systems}

In this section, we define kinematic systems and propose a framework in which to study them.
This section provides solutions to steps \ref{item:analysis_state} and \ref{item:analysis_kinematics} in the proposed observer design methodology \S\ref{sub:design_outline}.
Although this material is closely related to work on the modeling of mechanical systems \cite{marsden1999,bloch2003,bullo2004},
the focus on only the kinematics of the system and the importance of understanding the velocity measurements in formulating the observer problem, leads to new perspectives and warrants a careful development.

The goal of the section is to derive an input-state model $\dot{\xi} = f(\xi, v)$ for a general observer problem.
A key perspective is that the input $v(t)$ to such a model must itself be measured and modelling these measurements is a critical part of the observer design problem.
As a consequence, the structure of the system function $f$ cannot be assumed, but must be chosen to be compatible with the particular velocity measurements available.
The same `system' with different velocity measurements will yield a different input-state model.
This perspective is different from most treatments of the observer design problem where the starting position is an ordinary differential equation $\dot{\xi} = f(\xi, v)$ and the question of velocity measurement is not considered explicitly.
The advantage in the more abstract approach is that the role of the velocity measurement is made clear, an important point in \S\ref{sec:equivariant_representation} when invariant input extensions are discussed.
The reader who wishes to go directly to the material on symmetry where the model is assumed can skip straight to \S\ref{sec:equivariant_representation}.

%~~~~~~~~~~~~~~~~~~~~~~~~~~~~~~~~~~~~~~~~~~~~~~~~~~~~~~~~~~~~~~~~~~~~~~~~~~~~~
\subsection{Kinematic Systems and Kinematic Models}\label{sub:kinematic_models}

Consider a system with state $\xi \in \calM$ evolving on a real $n$-dimensional manifold $\calM$.
We begin from an abstract point of view, defining a general kinematic behaviour on $\tT \calM$ as a behaviour that satisfies the fundamental kinematic constraint \eqref{eq:abstract_kin}.

\begin{definition}\label{def:kinematic_system}
A \emph{system} on a smooth manifold $\calM$ is a triple $([0,\infty), \tT \calM , \calB)$ where $\calB$ (the behaviour \cite{1998_Willems}) specifies a subset of all trajectories $\eta_\xi : [0,\infty) \to \tT\calM$.
A \emph{kinematic system} is one where the base point evolution $\xi(t) = \pi(\eta_\xi(t))$ is time differentiable and the behaviour satisfies the kinematic constraint
\begin{align}
\ddt \xi(t) = \eta_\xi (t). \label{eq:abstract_kin}
\end{align}
\end{definition}

The kinematic constraint \eqref{eq:abstract_kin} is the minimum constraint for a trajectory $\eta_\xi(t) \in \tT \calM$ to make physical sense; it encodes the property that $\eta_\xi$ is the velocity $\dot{\xi}(t)$ of $\xi(t)$ along trajectories of $\calB$.
Many systems, such as non-holonomic systems, have additional constraints that further restrict the set of trajectories in the behaviour, while the behaviour of systems such as first order kinematics of rigid-body motion are only constrained by \eqref{eq:abstract_kin}.
An observer that works for trajectories in a general behaviour also works for trajectories in a sub-behaviour and the theory developed will often exploit this property by embedding the behaviour of a system of interest into a larger behaviour that has nicer symmetry properties.

\begin{definition}\label{def:linear_system_funtion}
A \emph{linear system function} $f : V \rightarrow \gothX(\calM)$ is a linear homomorphism from a linear vector space $V$
\begin{align}
f & : V \to \gothX(\calM), \quad\quad  v \mapsto f_v.
\label{eq:f_v}
\end{align}

A linear system function is said to be a representation of a kinematic system $([0,\infty), \tT \calM , \calB)$ if there exists an extended system
$([0,\infty), \tT \calM  \times V , \calB \times \calB_V)$ for a product behaviour $\calB \times \calB_V$ such that
\[
\eta_\xi(t) = f_{v(t)}(\pi(\eta_\xi(t))) = f_{v(t)}(\xi(t)) \quad \text{ for all } \quad (\eta_\xi(t), v(t)) \in \calB \times \calB_V.
\]
\end{definition}

A linear system function representation of a kinematic behaviour $\calB$ allows trajectories of $\calB$ to be characterised as solutions of an input-state system
$\dot{\xi} = f_{v(t)}(\xi(t))$,
for input trajectories $v(t) \in  V$ drawn from the behaviour $\calB_V$.

The power of the behavioural perspective is made clear in considering a general affine control system
\begin{align}
\dot{\xi} = f_0(\xi) + \sum_{j=1}^m f_j(\xi) u_j
\label{eq:control_affine_system}
\end{align}
for general inputs $u(t)$.
The implicit understanding in writing this equation is that the inputs $u_j$ for $j = 1, \ldots, m$ are arbitrary and the behaviour $\calB$ is the set of all solutions $(\xi(t),\dot{\xi}(t))$ of \eqref{eq:control_affine_system} for all possible input functions $u(t) = (u_1(t),\ldots,u_m(t))$ and initial conditions.
To find a linear system function representation of these kinematics one can  augment the input space to include a drift input $u_0 \in \R$ and define the linear system function  $f : \R^{m+1} \to \gothX(\calM)$ by
\begin{align}
f_u (\xi) := \sum_{j=0}^m f_j(\xi) u_j.
\label{eq:LSF_affine_control_system}
\end{align}
It is direct to see that this is a linear homomorphism satisfying the first requirement of Def.~\ref{def:linear_system_funtion}.
The behaviour $\calB_V$ includes all the input trajectories $u(t)$ associated with solutions of \eqref{eq:control_affine_system} along with an additional signal $u_0(t) \equiv 1$ for all time and $f_u$ is a linear system function for $\calB$.

\begin{remark}
The construction \eqref{eq:LSF_affine_control_system} is particularly important in dealing with second order kinematic systems.  For example, consider second order kinematics with configuration state $(x,v) \in \tT \R^n$
\begin{align*}
\dot{x} & = v \\
\dot{v} & = a.
\end{align*}
and with $a \in \tT_\xi \R^n \triangleq \R^n$ the measured acceleration.
Such systems always have a drift term $(\dot{x},\dot{v}) = (v,0)$ that is non-zero even if the acceleration measurement is zero.
The kinematic model for such a system is
\begin{align*}
\ddt \begin{pmatrix}
  x \\ v
\end{pmatrix}
=
\begin{pmatrix}
  v & 0 \\ 0 & \Id_n
\end{pmatrix}
\begin{pmatrix}
u_0
 \\  u
\end{pmatrix}
\end{align*}
for general $(u_0,u) \in \R^{n+1}$.
Trajectories of the original system are obtained by choosing $(u_0, u) = (1, a)$.
This construction generalises in a straightforward manner to systems on the tangent bundle $\tT\calM$ of more general manifolds $\calM$.
\end{remark}

The linear system function $f : V \to \gothX(\calM)$ defines a linear subspace $\image f \subset \gothX(\calM)$ of the vector space of vector fields over $\calM$.
While $V$ itself is just a linear vector space, the subspace $\image f$ contains structure associated with its embedding as a subspace of $\gothX(\calM)$.
Since $f$ is a linear map then, at least assuming trivial kernel, one can identify the vector space $V$ with its image $V_f := \image f$.
There is a natural system function $f^\imath : V_f \to \gothX(\calM)$ given by the identity inclusion $f^\imath (f_v) := f_v$.
That is, if $f : V \to \gothX(\calM)$ is a linear system function representation of a behaviour $\calB$,  and $f$ has trivial kernel, then $f^\imath : V_f \to \gothX(\calM)$ is also a linear system function representation of the behaviour $\calB$.
The reverse implication is also true, for any subspace $V \in \gothX(\calM)$ then $f^\imath : V \to \gothX(\calM)$ is a linear system function representation of a behaviour defined by the solutions of $\dot{\xi} = f^\imath_{v(t)} (\xi)$ for any continuous choice of input $v(t) \in V$.
It follows that one can think of linear system function representations as subspaces $V \subset \gothX(\calM)$ along with the identity inclusion $f^\imath$.
It is important to note that this discussion and Def.~\ref{def:linear_system_funtion} does not require that the input vector space $V$ is finite dimensional.

\begin{definition}\label{def:velocity_output}
Let $\tT \calM$ be the tangent bundle of a state manifold $\calM$.
A \emph{velocity output} is a smooth map $g : \tT \calM \rightarrow \vecV$, into a finite-dimensional vector space $\vecV$, for which the map $g$ is linear on fibres of the bundle.  That is for any $\xi \in \calM$ and for all $k_1, k_2 \in \R$ and $\eta_\xi, \mu_\xi \in \tT_\xi \calM $ then
\[
g( k_1 \eta_\xi + k_2 \mu_\xi) = k_1  g(\eta_\xi)  + k_2 g(\mu_\xi).
\]
\end{definition}

Velocity measurements for physical systems are obtained from physical measurement devices whose own dynamics interact with the physical state of the system of interest, isolating and extracting part of the velocity state in a way that can then be measured and quantified as a real number.
For example, the angle of deflection of an oscillating mass in a MEMS rate gyroscope is linearly correlated to the angular velocity of the rigid-body to which it is attached, and can be measured as a voltage change within the device.
The doppler shift of the carrier frequency of GPS corresponds to the linear velocity of the receiver in the direction of the satellite, and can be measured by FFT analysis of the received signal, \textit{etc}.
In all cases each individual measurement is a scalar real number.
This model of multiple scalar measurements underlies all velocity sensor suites.

Real measurement devices have limited range, have non-linearities in the measurement process, as well as being susceptible to noise and bias.
These issues are critical to the real world performance of observers and should be considered at the appropriate point in the design process.
However, in the present development of a theoretical framework for equivariant observer design, we will consider noise free, or ideal, values of the physical velocity measured by the sensor device.
We use the term \emph{velocity output} to distinguish this ideal signal from a physical velocity measurement that would be corrupted by noise and sensor characteristics.

Physical velocities are linear quantities and can always be scaled and added.
Requiring linearity of the velocity output $g$ on the fibres of $\tT \calM$ enforces that the sensor does not destroy the natural linear structure of physical velocity.
Each scalar velocity output corresponds to velocity of the true system in one direction and can be positive or negative.
It follows that an idealised scalar velocity output maps into the whole of the real line $\R$.
Combining multiple outputs together naturally leads to a vector velocity output $g : \tT \calM \to \R^p =: \vecV$, where $p$ is the number of outputs and $\vecV$ is the notation that we introduce to denote the finite-dimensional input velocity vector space.

\begin{definition}\label{def:systemfunction_compatible}
Let $g : \tT\calM \to \vecV$ be a velocity output (Def.~\ref{def:velocity_output}).
A linear system function $f : \vecV \to \gothX(\calM)$ (Def.~\ref{def:linear_system_funtion}) is said to be \emph{compatible} with $g$ if
\begin{align}
\eta_\xi = f_{g(\eta_\xi)}(\xi) \label{eq:compatible_velocity_output}
\end{align}
for all $\xi\in\calM$ and $\eta_\xi \in \tT_\xi\calM$.
\end{definition}

Not every velocity output $g : \tT\calM \to \vecV$ admits a compatible linear system function $f : \vecV \to \gothX(\calM)$.
The following property provides a necessary and sufficient condition on a velocity output that for the existence of a compatible linear systems function $f$.

\begin{definition}\label{def:complete_output_g}
Let $g : \tT\calM \to \vecV$ be a velocity output (Def.~\ref{def:velocity_output}).
The output $g$ is termed \emph{complete} if, for all $\xi \in \calM$, the restricted map $g : \tT_\xi \calM \to \vecV$ is injective.
\end{definition}

Definition~\ref{def:complete_output_g} ensures that $g$ captures the full velocity information of the system at each point $\xi \in \calM$.

\begin{lemma}\label{lem:complete_output_g}
Let $g : \tT\calM \to \vecV$ be a velocity output (Def.~\ref{def:velocity_output}).
Then there exists a linear system function $f : \vecV \to \gothX(\calM)$ that is compatible with $g$ (Def.~\ref{def:systemfunction_compatible}) if and only if $g$ is complete (Def.~\ref{def:complete_output_g}).
\end{lemma}

\begin{proof}
This follows from the fact that a linear map has a left inverse if and only if it is injective. Apply this fact point-wise for every $\xi\in\calM$. Smoothness of $f$ follows from smoothness of $g$.
\end{proof}

For a complete velocity output with compatible linear system function the abstract kinematics \eqref{eq:abstract_kin} of a kinematic system $\calB$ are fully characterised by the two constraints $\dot{\xi}(t) = f_{v(t)} (\xi)$, and $v(t) = g(\eta_\xi(t))$.
In particular, trajectories in $\calB$ are solutions of the ordinary differential equation
\begin{align}
\dot{\xi}(t) = f_{v(t)}(\xi(t)), \quad \xi(0) \in \calM  \label{eq:ODE_xi}
\end{align}
and $f$ is a linear system function representation (Def.~\ref{def:linear_system_funtion}) of $\calB$.

\begin{remark}
If the velocity distribution of the behaviour does not span the full tangent space $\tT_\xi \calM$ Definition \ref{def:complete_output_g} can be weakened to be injective on the distribution of accessible velocities of the system.
Examples of such systems are non-holonomic systems and second order kinematic systems.
In such cases Definition \ref{def:systemfunction_compatible} can be adapted to only require \eqref{eq:compatible_velocity_output} to hold for $\eta_\xi$ on the distribution of accessible velocities of the system and the linear system function is said to be \emph{compatible with $g$ on the behaviour $\calB$}.
A necessary and sufficient condition for the existence of such a linear systems function is that the velocity output is injective on the distribution of accessible velocities of the system and such an output is termed \emph{complete on the behaviour $\calB$}.
\end{remark}

Let $\ker f \subset \vecV$ denote the kernel of $f : \vecV \rightarrow \gothX(\calM)$,
\[
\ker f = \{ v \in \vecV \;|\;  f_v(\xi) = 0 \text{ for all } \xi \in \calM \}.
\]
The kernel of $f$ may be non-zero if there are linearly dependent velocity outputs on the vehicle, for example multiple inertial measurement units on a rigid body.
There are often very good reasons for building redundancy into velocity measurement systems, however, managing this redundancy must be handled outside the mathematical framework that we develop in this paper and we will factor out such dependency.
The function $\bar{f} : \vecV / \ker f \rightarrow \gothX (\calM)$ is well defined since $f_{v + k} = f_v + f_k = f_v$ for all $k \in \ker f$.
The quotient $\vecV / \ker f$ is itself a linear vector space and without loss of generality the analysis can be restricted to the case where $\ker f$ is trivial; that is, injective linear system functions.

\begin{lemma}\label{lem:injective_system_function}
Let $g : \tT\calM \to \vecV$ be a complete velocity output and let $f : \vecV \to \gothX(\calM)$ be a linear system function that is compatible with $g$. Then the velocity output $\bar{g} : \tT\calM \rightarrow \vecV / \ker f$, $\bar{g}(\eta):=g(\eta)+\ker f$ is complete and the injective linear system function $\bar{f} : \vecV / \ker f \rightarrow \gothX (\calM)$ is compatible with $\bar{g}$.
\end{lemma}

\begin{proof}
Let $\xi\in\calM$ and let $\eta_\xi,\mu_\xi\in\tT_\xi\calM$ with $\eta_\xi\not=\mu_\xi$.
Assume, to arrive at a contradiction, that $\bar{g}(\eta_\xi)=\bar{g}(\mu_\xi)$.
Then $k:=g(\eta_\xi)-g(\mu_\xi)\in \ker f$ and $0=f_k(\xi)=f_{g(\eta_\xi)}(\xi)-f_{g(\mu_\xi)}(\xi)=\eta_\xi-\mu_\xi$ by \eqref{eq:compatible_velocity_output}, a contradiction. It follows that $\bar{g} : \tT_\xi\calM \to \vecV / \ker f$ is injective and $\bar{g}$ is a complete velocity output. Compatibility of $\bar f$ follows from
$\bar{f}_{\bar{g}(\eta_\xi)}(\xi)=\bar{f}_{g(\eta_\xi)+\ker f}(\xi)=f_{g(\eta_\xi)}(\xi)=\eta_\xi$.
\end{proof}

\begin{definition}\label{def:kinematic_model}
Assume that there is a kinematic system $([0,\infty), \tT\calM, \calB)$ (Def.~\ref{def:kinematic_system}) with smooth trajectories.
A \emph{kinematic model} $f: \vecV \to \gothX(\calM)$, is an injective linear system function $f$ (Def.~\ref{def:linear_system_funtion}) on a finite-dimensional vector space $\vecV$ that is compatible (Def.~\ref{def:systemfunction_compatible}) with a complete (Def.~\ref{def:complete_output_g}) velocity output $g : \tT \calM \to \vecV$ (Def.~\ref{def:velocity_output}).
\end{definition}

For the kinematic model $f : \vecV \to \gothX(\calM)$, the notation $v \mapsto f_v$ emphasises that the linear system function is thought of as a linear homomorphism, and maps into the vector space of smooth vector fields. However, when $f_v \in \gothX(\calM)$ is evaluated at a particular point $\xi \in \calM$ we will use the notation
\[
f : \calM \times \vecV \to \tT \calM, \quad\quad f(\xi,v) := f_v(\xi)
\]
interchangeably with the notation $f_v$ to emphasise the classical input-state systems structure of the representation.

By construction, a kinematic model representation is global.
It differs from classical geometric control systems formulations \cite{marsden1999,bloch2003,bullo2004} in that the bundle $\calM \times \vecV$ on which the input and state are defined is trivial by construction.
This property is tied to the particular nature of the velocity outputs and deserves further discussion.
The vector space $\vecV$ can be, and often is, of higher dimension than the manifold $\calM$ depending on the availability of different velocity sensor systems.
Using multiple velocity outputs is critical in overcoming structural limitations associated with velocity representation for systems on non-trivial state-space bundles.
For example, consider the velocity of a point on a sphere.
There is no non-degenerate smooth velocity map from $\tT \Sph^2 \to \R^2$
\cite{1912_Brouwer_MA,1979_Eisenberg_AMM} and this is a fundamental obstruction to deploying only two velocity sensors to measure the velocity of the point at every point on the sphere.
For example, consider using some sort of physical azimuth and elevation velocity output devices that measure the velocity of angles describing the position of the point on the sphere.
If only two such outputs are considered, there will always be singular points where one of the elevation or azimuth velocity outputs is in gymbal lock and measures zero for all instantaneous motions.
At such points the combination of the two velocity outputs drops rank and would fail to be complete (Def.~\ref{def:complete_output_g}).
This singularity can be overcome by using an additional elevation or azimuth velocity output taken with respect to a different axis.
That is, three velocity outputs are required to provide a complete velocity output for kinematics on a two-dimensional manifold $\Sph^2$.
An alternative sensor modality for this system is to mount a 3DOF strap-down gyro to the physical point on the sphere providing the full three-dimensional angular velocity associated with a physical frame moving that is rigidly attached to the point.
These three velocity outputs are also a complete velocity output for the point kinematics.
The two different sensor suites (three or more azimuth/elevation outputs versus the three angular velocities of a frame attached to the point) will yield different kinematic models and lead to different observer designs.

For the same behaviour $\calB$, different velocity outputs $g$ will call for different linear system functions $f$ and lead to different kinematic models.
The same behaviour $\calB$ with different sensors will lead to a different observer design problem, and ultimately to a different observer.

\begin{theorem}\label{th:LinearSystemFunctionRep_exists}
All kinematic systems (Def.~\ref{def:kinematic_system}) admit a kinematic model (Def.~\ref{def:kinematic_model}).
\end{theorem}

\begin{proof}
For any smooth manifold $\calM$ of dimension $n$ then the Whitney embedding theorem \cite{1993_Adachi} states that there is an embedding map $\imath : \calM \hookrightarrow \R^{n_\star}$ with $n_\star \geq 2n$.
For any element $\eta_\xi \in \tT_\xi \calM$ there exists a curve $\xi(t)$ in $\calM$ with $\dot{\xi}(0) = \eta_\xi$.
Define $g : \tT\calM \to \R^{n_\star}$
\[
g(\eta_\xi) = \left. \ddt \right|_{0}\imath \circ \xi (t) \in \R^{n_\star}.
\]
It is straightforward to see that this defines a complete velocity output into a finite dimensional vector space $\R^{n_\star}$.
By Lemma \ref{lem:complete_output_g} there exists a linear system function
$f : \R^{n_\star} \to \gothX(\calM)$ that is compatible with $g$. The result now follows from Lemma \ref{lem:injective_system_function}.
\end{proof}

Clearly this result is of theoretical interest only, since construction of embedding maps is not a practical manner in which to generate velocity outputs.
However, Theorem \ref{th:LinearSystemFunctionRep_exists} serves to emphasise that the inherent reformulation of the input-state structure of the system onto the trivial bundle $\calM \times \vecV$ with the linear system function structure is not a restriction on the class of systems that can be modelled.

In conclusion, we distinguish strongly between a \emph{kinematic system} consisting of a behaviour $\calB$ satisfying \eqref{eq:abstract_kin} on $\tT\calM$, a general \emph{linear system function representation} $f : V \to \gothX(\calM)$ of the kinematic system, and a \emph{kinematic model} that is an injective linear system function $f : \vecV \to \gothX(\calM)$ that is compatible with a complete velocity output $g : \tT \calM \to \vecV$ on a finite-dimensional vector space $\vecV$ for which the trajectories satisfy the ODE system model \eqref{eq:ODE_xi}.

%============================================================================
\section{Equivariant Systems}\label{sec:equivariant_representation}
%~~~~~~~~~~~~~~~~~~~~~~~~~~~~~~~~~~~~~~~~~~~~~~~~~~~~~~~~~~~~~~~~~~~~~~~~~~~~~

This section considers the question of symmetry and develops the underlying theory of equivariant systems.
We assume that a kinematic model for the system (Def.~\ref{def:kinematic_model}) is available for which the state-space $\calM$ is a homogeneous space.
The associated linear system function is a representation for the behaviour and is the primary structure that is used in the development of the equivariant systems theory.
The section provides solutions to steps \ref{item:analysis_symmetry}, \ref{item:analysis_input_sym} and \ref{item:analysis_lift} for the proposed observer design methodology \S\ref{sub:design_outline}.

%~~~~~~~~~~~~~~~~~~~~~~~~~~~~~~~~~~~~~~~~~~~~~~~~~~~~~~~~~~~~~~~~~~~~~~~~~~~~~
\subsection{Symmetry}\label{sub:symmetry}

\begin{definition}\label{def:symmetry_group}
Let $\grpG$ be a Lie group and $\phi : \grpG \times \calM \to \calM$ be a transitive smooth action on a smooth manifold $\calM$.
A linear system function (Def.~\ref{def:linear_system_funtion}) on $\calM$ is said to have \emph{homogeneous state}.
\end{definition}

A linear system function representation of a kinematic system with homogeneous state imposes no additional constraints on the system kinematics other than that the state space $\calM$ is a homogeneous space.
In particular, we make no assumption that the linear system function $f : V \to \gothX(\calM)$ has any equivariance or symmetry properties to begin with, allowing us to consider a very wide range of examples.

A choice as to the handedness, left or right, of the group action $\phi$ must be made during the formulation of the observer design problem.
In practice, there are accepted models for robotic systems in the literature that go along with commonly accepted understanding for the meaning of certain state variables.
Associated with these models, there are natural choices of handedness in order that the various system states carry the standard physical interpretation.
Right-handed invariance is the most natural representation for physical systems with body-fixed sensor systems.
This includes the majority of robotic applications of interest where  vehicles have onboard sensor systems.
This includes attitude estimation, pose estimation, velocity-aided-attitude, SLAM, VO, VIO.
Left-handed invariance is more natural for a situation where a ground based system estimates state of an observed vehicle moving through its sensor field.
Since the authors come from a robotics perspective, where the majority of applications involve embedded sensor systems mounted on moving robots, we choose to use the right-handed invariance to develop the theory.

\begin{remark}
If desired, the choice of handedness can be reversed by re-defining the group multiplication.
More precisely, for a group $\grpG$ the operation $\odot \colon \bar{\grpG}\times\bar{\grpG}\rightarrow\bar{\grpG}$, $A\odot B:=BA$ turns a copy $\bar{\grpG}$ of the \emph{set} $\grpG$ into a group that is isomorphic to the \emph{group} $\grpG$ via $A\mapsto A^{-1}$.
This simple transformation of the group reverses the handedness of the induced action $\ob{\phi} :  \bar{\grpG} \times \calM \to \calM$, $\ob{\phi}(A,\xi) := \phi(A,\xi)$.
All results that follow have direct analogues in the opposite handedness with suitable development.
\end{remark}

\begin{assumption}\label{ass:right_handed}
The group action $\phi$ in Definition \ref{def:symmetry_group} is right-handed.
\end{assumption}

A group action $\phi$ of $\grpG$ on $\calM$ induces a group action of $\grpG$ on $\gothX(\calM)$ that we will term the \emph{induced action on vector fields}.
The underlying construction for the pushforward of a vector field associated to a diffeomorphism on a manifold is known \cite{2003_Lee}, however, the generalisation of this to a Lie group action, although direct, does not appear to be well documented in the literature.

\begin{lemma}\label{lem:d_star_phi}
\textbf{[Induced action on vector fields]}
Let $\phi : \grpG \times \calM \rightarrow \calM$ be a transitive (right) group action on $\calM$.
Define a smooth map $\td_\star \phi  : \grpG \times \gothX(\calM) \rightarrow \gothX(\calM)$, by
\begin{align}
\td_\star \phi (Z, f) := \td \phi_Z f \circ \phi_{Z^{-1}}.
\label{eq:overline_dphi}
\end{align}
Then $\td_\star \phi$ is a Lie-algebra group action on $\gothX(\calM)$.
That is, $\td_\star \phi$ is a group action for which each mapping
$\td_\star \phi_Z : \gothX(\calM) \rightarrow \gothX(\calM)$ is a Lie-algebra homomorphism.
\end{lemma}

\begin{proof}
It is straightforward to verify that $\td_\star \phi$, a parameterized family of pushforwards of vector fields for each group element, is well defined.
Note that
\[
\td_\star \phi (\Id,f) =
\td \phi_{\Id} f \circ \phi_{\Id^{-1}} = f.
\]
For $X, Y \in \grpG$ compute
\begin{align*}
\td_\star \phi (Y, \td_\star \phi (X, \mu)))
& =
\td \phi_Y \left(\td \phi_X \mu \circ \phi_{X^{-1}}\right) \circ \phi_{Y^{-1}}
 =
\td \phi_{XY} \mu \circ \phi_{(XY)^{-1}}
 =
\td_\star \phi (XY,\mu).
\end{align*}
The Lie-algebra homomorphism property follows from the fact that vector field Lie-brackets are equivariant under transformation by smooth tangent maps.
\end{proof}

%~~~~~~~~~~~~~~~~~~~~~~~~~~~~~~~~~~~~~~~~~~~~~~~~~~~~~~~~~~~~~~~~~~~~~~~~~~~~~
\subsection{Equivariant Kinematics and the Extended Input Space}\label{sub:extended_kinematics}

In this section, we show that for a linear system function $f : V \to \gothX(\calM)$ with homogeneous state, then there is a natural extension of the input space $\calV \supseteq V$ and an extended linear system function $\calf : \calV \to \gothX(\calM)$ that is equivariant under the group action $\phi$.
We begin by defining an equivariant linear system function representation on a general input vector space $V$ and then go on to develop the input extension.

\begin{definition}\label{def:equivariant_model}
Consider a linear system function $f : V \to \gothX(\calM)$ (Def.~\ref{def:linear_system_funtion}) with homogeneous state (Def.~\ref{def:symmetry_group}).
This function is termed \emph{equivariant} if there exists a vector space group action
\begin{align}
\psi & : \grpG \times V \rightarrow V \label{eq:psi}
\end{align}
such that for all $X \in \grpG$, $\xi \in \calM$, $v \in V$ then
\begin{align}
\td \phi_X f(\xi, v) = f(\phi_X(\xi),\psi_X(v) ).
\label{eq:psi_compatible}
\end{align}
\end{definition}

Consider a linear system function $f:V \to \gothX(\grpG)$ with homogeneous state and assume that there is a vector space group action $\psi : \grpG \times V \to V$ that is well defined on the input space $V$ for which the resulting system is equivariant (Def.~\ref{def:equivariant_model}).
Then
\begin{align}
f_{\psi_A(v)} (\xi)
& =  f_{\psi_A(v)} (\phi_A \phi_{A^{-1}}(\xi))\notag \\
& = \td \phi_A f_v (\phi_{A^{-1}}(\xi)) \label{eq:Fpsi:2} \\
%& = \td \phi_A f_v \circ \phi_{A^{-1}}(\xi) \notag \\
& = \td_\star \phi _A f_v (\xi) \label{eq:Fpsi:3}
\end{align}
where \eqref{eq:Fpsi:2} follows from \eqref{eq:psi_compatible} and \eqref{eq:Fpsi:3} follows from \eqref{eq:overline_dphi}.
In particular, the input group action $\psi$ is uniquely determined by the induced action $\td_\star \phi$ (Lemma~\ref{lem:d_star_phi}).
For $A \in \grpG$ then $\td_\star \phi_A $ is always well defined as a map $\gothX(\calM) \to \gothX(\calM)$.
Thus, \eqref{eq:Fpsi:3} can be viewed as stating that $\td_\star \phi $ maps elements of $\image f \subset \gothX(\calM)$ into $\image f$, in particular $f_v \mapsto f_{\psi_A(v)}$, that is the induced action $\td_\star \phi  : \grpG \times \image f \to \image f$ is closed on $\image f \subset \gothX(\calM)$.

Looking at the above discussion from the opposite perspective, then given a linear system function $f : V \to \gothX(\calM)$ on a general input vector space $V$ one can extend the input space of $f$ to include the closure of $\image f$ under action by $\td_\star \phi $.

\begin{definition}\label{def:input_ext}\textbf{[Equivariant Input Extension]}
Let $f : V \to \gothX(\calM)$ be a linear system function with homogeneous state.
Define
\begin{align}
\calV = \spn \{ \td_\star \phi _A f_v \;|\; v \in V, A \in \grpG \} \subset \gothX(\calM) \label{eq:input_extension}
\end{align}
to be the smallest subspace of $\gothX(\calM)$ generated by the image of $\td_\star \phi $ applied to $\image f$.
\end{definition}

Note that any element in $\calV$ can be written
\begin{align*}
\calf_u =   \sum_{i = 1}^K \td_\star \phi _{A_i} f_{v_i}
\end{align*}
for $K \in \N$ and some collection $\{A_i\} \in \grpG$ and $\{v_i\} \in V$.
Although each element $f_u \in \calV$ can be seen directly as an element of $\gothX(\calM)$ we will continue to index elements using lower case letters $u, v, w \in \calV$ via the correspondence $u \leftrightarrow \calf_u$.
This notation emphasises the vector space nature of $\calV$ and the role of its elements as inputs.
If $v \in V$ then $\calf_v = f_v$ corresponds to an element in the image of the original linear system function by construction.
With this notation, the vector space $\calV$ can be viewed as an input space for an extended system by inclusion
\begin{align}
\calf : \calV \to \gothX(\calM). \quad\quad u \mapsto \calf_u  :=  \sum_{i = 1}^K \td_\star \phi _{A_i} f_{v_i}
\label{eq:calf}
\end{align}
for a suitable collection of $\{v_i\} \in V$ and $\{A_i\} \in \grpG$.
The resulting function is clearly linear in the input $u$ and satisfies the criteria to be a linear system function.

\begin{theorem}\label{th:equivariant_extension}
Let $f : V \to \gothX(\grpG)$ be a linear system function with homogeneous state and let $\calV$ be its equivariant input extension (Def.~\ref{def:input_ext}).
The linear system function $\calf : \calV \rightarrow \gothX(\calM)$ generated by the correspondence \eqref{eq:calf} is equivariant with group action $\psi : \grpG \times \calV \to \calV$ defined by
\begin{align}
\calf_{\psi_A(w)} := \td_\star \phi _A \calf_w
\label{eq:equivariant_psi}
\end{align}
for all $A\in\grpG$ and $\calf_w\in\calV$.
\hfill$\square$
\end{theorem}

\begin{proof}
For any $\calf_w = \sum \td_\star \phi _{A_i} f_{v_i} \in \calV$ then
\begin{align*}
\td_\star \phi _B \calf_w & =\td_\star \phi _B \sum \td_\star \phi _{A_i} f_{v_i}
%\\
%& =\sum \td_\star \phi _B  \td_\star \phi _{A_i} f_{v_i} \\
 = \sum \td_\star \phi _{A_i B} f_{v_i} \in \calV.
\end{align*}
That is, $\td_\star \phi $ is closed on $\calV$.
Since $\td_\star \phi $ is a group action then $\td_\star \phi_I=\id$ and $\td_\star \phi _A \td_\star \phi _B = \td_\star \phi _{BA}$, and hence $\psi_I=\id$ and $\psi_A\psi_B=\psi_{BA}$ for all $A,B\in\grpG$. Hence $\psi$ is a group action. Linearity of $\psi_A$ for each $A\in\grpG$ follows from linearity of $\td_\star \phi _A$, that is $\psi$ is a vector space group action.
Recalling \eqref{eq:equivariant_psi} then
\[
\td \phi_A \calf_w (X) = \calf_{\psi_A(w)}(\phi_A (X))= \calf( \phi_A (X), {\psi_A(w)})
\]
for $w \equiv \calf_w \in \calV$.
This proves equivariance.
\end{proof}

%~~~~~~~~~~~~~~~~~~~~~~~~~~~~~~~~~~~~~~~~~~~~~~~~~~~~~~~~~~~~~~~~~~~~~~~~~~~~~~~~~~~~~~~~~~~~~~~~~~~~~~~
\subsection{System lifts on the Symmetry Group}\label{sub:system_lift}

\begin{definition}\label{lem:system_lift}
Consider a linear system function $f : V \to \gothX(\calM)$ with homogeneous state (Def.~\ref{def:symmetry_group}).
A smooth map $\Lambda : \calM \times V \rightarrow \gothg$, linear in the $V$ variable, such that
\begin{align}
\td \phi_{\xi} \Lambda(\xi,v) = f(\xi,v) \label{eq:project_Lambda}
\end{align}
for all $\xi\in\calM$ and $v\in V$ is termed a \emph{lift}.
\end{definition}

Since $\calM$ is a homogeneous space the group action $\phi$ is transitive and the partial map $\phi_\xi$ is surjective for any $\xi \in \calM$.
Therefore, $\tD \phi_\xi(\id)$ is full rank and varies smoothly and there must exist at least one smooth right-inverse $\tD \phi_\xi(\id)^\dagger$.
For such a right inverse, define
$\Lambda(\xi, v) := \tD \phi_\xi(\id)^\dagger \cdot f(\xi, v)$.
It is straightforward to verify that this $\Lambda$ satisfies the requirements of Definition \ref{lem:system_lift} and consequently that a lift exists for any linear system function with homogeneous state.
If the action is free then $\tD \phi_\xi(\id)$ is invertible and the lift is unique.

A lift $\Lambda : \calM \times V \to \gothg$ provides the necessary structure to construct a \emph{lifted system} on the symmetry group.
In order to make this construction, it is necessary to first pick a reference point $\mr{\xi} \in \calM $ that we term the \emph{origin point}.
The role of this point is to provide an origin for a global coordinate parametrization $\phi_{\mr{\xi}} : \grpG \rightarrow \calM$ of the state-space $\calM$ by the symmetry group $\grpG$.
For a lift function  $\Lambda (\xi,v)$ (Lemma \ref{lem:system_lift}) then the \emph{lifted system} is given by
\begin{align}
\dot{X} = \td L_X \Lambda(\phi_{\mr{\xi}}(X), v), \quad X(0) \in \grpG \label{eq:lifted_system}
\end{align}
Solutions to the lifted system evolve on the Lie-group $\grpG$ and project down onto trajectories of the kinematic system \eqref{eq:ODE_xi} via the map $\phi_{\mr{\xi}}$.

\begin{lemma}\label{lem:project_system}
Let $f : V \rightarrow \gothX(\calM)$ be a linear system function representation (Def.~\ref{def:linear_system_funtion}) of a kinematic system $([0,\infty),T \calM, \calB)$ (Def.~\ref{def:kinematic_system}) with  extended behaviour  $\calB \times \calB_V$.
Consider an input signal $v(t) \in V$ for $t \geq 0$ drawn from the behaviour $\calB_V$.
Let $X(t)$ denote the solution to the lifted system \eqref{eq:lifted_system} for input $v(t)$.
If the initial condition $X(0) \in \grpG$ satisfies
\[
\phi_{\mr{\xi}}(X(0)) = \xi(0)
\]
then
\[
\phi_{\mr{\xi}}(X(t)) = \xi(t) \textrm{ for all }t \geq 0,
\]
where $\xi(t)$ is the solution of \eqref{eq:ODE_xi}; that is, $(\xi(t),\dot{\xi}(t)) \in \tT\calM$ lies in the behaviour $\calB$.
\end{lemma}

\begin{proof}
Set $\xi(t) = \phi_{\mr{\xi}}(X(t))$ for $X(t)$ a solution to \eqref{eq:lifted_system} with $\phi_{\mr{\xi}}(X(0)) = \xi(0)$.
Observe that
\begin{align*}
\dot{\xi}(t) & = \ddt \phi_{\mr{\xi}}(X(t)) = \td \phi_{\mr{\xi}} \dot{X}(t) = \td \phi_{\mr{\xi}} \td L_X \Lambda(\phi_{\mr{\xi}}(X), v)
%\td \phi_{\mr{\xi}} \td L_X \Lambda(\phi_{\mr{\xi}}(X), v)
%& = \td \phi_{\mr{\xi}} \td L_X \Lambda(\xi, v) \\
%& = \td ( \phi_{\mr{\xi}}\circ L_X ) \Lambda(\xi, v) \\
 = \td \phi_{\xi} \Lambda(\xi, v)
 = f(\xi,v).
\end{align*}
That is $\xi(t)$ satisfies \eqref{eq:ODE_xi}.
Since the initial conditions match, the result follows from uniqueness of solutions to \eqref{eq:ODE_xi}.
\end{proof}

The lifted system \eqref{eq:lifted_system} is associated to a linear system function $F : V \to \gothX(\grpG)$,
\begin{align}
F(X,v) := \td L_X \Lambda(\phi_{\mr{\xi}}(X), v).
\label{eq:lifted_linear_systems_function}
\end{align}
This is not formally a linear system function representation of the original system according to Definition \ref{def:linear_system_funtion} since the signal space of the associated behaviour is $\tT \grpG$.
We will term this a \emph{lifted linear system function}.

%~~~~~~~~~~~~~~~~~~~~~~~~~~~~~~~~~~~~~~~~~~~~~~~~~~~~~~~~~~~~~~~~~~~~~~~~~~~~~~~~~~~~~~~~~~~~~~~~~~~~~~~
\subsection{Equivariant lifts}\label{sub:equivariant_lift}

If a linear system function is equivariant then it seems reasonable that the lift function (Def.~\ref{lem:system_lift}) and the lifted system \eqref{eq:lifted_system} will also have symmetric structure.
We will use notation $\calf : \calV \to \gothX(\calM)$ in this section to correspond to the common situation where the equivariant linear system function is generated by extension of a linear system function $f : V \to \gothX(\calM)$.

\begin{definition}\label{def:equivarian_lift}
Consider an equivariant linear system function $\calf : \calV \to \gothX(\calM)$.
A lift function $\Lambda$ (Def.~\ref{lem:system_lift}) for $\calf$ is termed \emph{equivariant} if
\begin{align}
\Ad_{X^{-1}} \Lambda(\xi,v) = \Lambda(\phi_X(\xi),\psi_X(v))
\label{eq:equivariant_infinitesimal_lift}
\end{align}
for all $\xi \in \calM$, $v \in \calV$ and $X \in \grpG$.
\end{definition}

The lifted linear system function \eqref{eq:lifted_linear_systems_function} associated with an equivariant lift is equivariant on $\grpG$.

\begin{lemma}\label{lem:equivariance_lifted_system}
Consider an equivariant linear system function $\calf : \calV \to \gothX(\calM)$  (Def~\ref{def:equivariant_model}).
Let $\Lambda$ be a lift function \eqref{lem:system_lift} for $\calf$ and assume it is equivariant (Def. \ref{def:equivarian_lift}).
Then the lifted linear system function $\calF := F$  \eqref{eq:lifted_linear_systems_function} is equivariant with respect to the group action $R$ (right translation) and input action $\psi$.
That is
\begin{align}
\td R_Z \calF(X,v)
= \calF(XZ ,\psi_Z(v))
\label{eq:equivariance_lifted_system}
\end{align}
for all $X,Z\in\grpG$ and $v\in\calV$.
\end{lemma}

\begin{proof}
One computes
\begin{align*}
\td R_Z \calF(X,v) & = \td R_Z (\td L_X \Lambda(\phi_{\mr{\xi}}(X),v) )
%= X \Lambda (\phi_{\mr{\xi}}(X),v) Z
% & =  X Z Z^{-1} \Lambda(\phi_{\mr{\xi}}(X),v)Z  \\
 =  X Z \Ad_{Z^{-1}} \Lambda(\phi_X(\mr{\xi}),v) \\
&  =  R_Z X  \Lambda (\phi_Z \phi_X(\mr{\xi}), \psi_Z(v) )
 %=  R_Z X  \Lambda ( \phi_{XZ}(\mr{\xi}), \psi_Z(v) )
 =  R_Z X  \Lambda ( \phi_{\mr{\xi}}(R_Z(X)), \psi_Z(v) ) \\
& =  \calF(R_Z(X),\psi_Z(v) )
\end{align*}
which concludes the proof.
\end{proof}

An equivariant lift is uniquely defined by its value at a single point $\mr{\xi} \in \calM$.
In particular, for $\mr{\xi} \in \calM$ fixed and $X \in \grpG$ such that $\phi_X(\mr{\xi}) = \xi$ then
\[
\Lambda(\xi,v) = \Lambda(\phi_X(\mr{\xi}),v) = \Ad_{X^{-1}} \Lambda(\mr{\xi},\psi_{X^{-1}}(v)).
\]
The following lemma provides a characterisation of an equivariant lift.

\begin{lemma}\label{lem:equivariant_lambda}
Consider an equivariant linear system function $\calf : \calV \to \gothX(\calM)$  (Def~\ref{def:equivariant_model}).
Fix an origin $\mr{\xi} \in \calM$.
A lift function $\Lambda$ for $\calf$ is equivariant if and only if there exists a map
$\Lambda_{\mr{\xi}} : V \rightarrow \gothg$ such that for all $v\in\calV$
\begin{align}
\td \phi_{\mr{\xi}} \Lambda_{\mr{\xi}}(v) & = \calf(\mr{\xi},v), \label{eq:velocity_lift_xi0} \\
\Ad_{S^{-1}}\Lambda_{\mr{\xi}}(v) & = \Lambda_{\mr{\xi}}(\psi_S(v)), \quad \textrm{\emph{for all} } S \in \stab_\phi (\mr{\xi}).  \label{eq:stab_equivaraince}
\end{align}
\end{lemma}

\begin{proof}
Firstly, we prove ``only if''.
If $\Lambda$ is an equivariant lift then define $\Lambda_{\mr{\xi}} (v) := \Lambda(\mr{\xi},v)$ and note
\begin{align*}
\Ad_{S^{-1}} \Lambda_{\mr{\xi}}(v) & = \Ad_{S^{-1}} \Lambda (\mr{\xi}, v)
 = \Lambda (\phi_S (\mr{\xi}),\psi_S(v))
 = \Lambda (\mr{\xi},\psi_S(v))
 = \Lambda_{\mr{\xi}}(\psi_S(v)).
\end{align*}

To prove ``if'' define a function
\begin{align}
\Lambda(\xi,v)  = \Lambda(\phi_X(\mr{\xi}),v) := \Ad_{X^{-1}} \Lambda_{\mr{\xi}}(\psi_{X^{-1}}(v))
\label{eq:lambda_construct_one}
\end{align}
where $X \in \grpG$ is any element such that $\xi = \phi_X(\mr{\xi})$.
First, it is necessary to verify that $\Lambda$ is a well defined map since \eqref{eq:lambda_construct_one} involves an element $X \in \grpG$ that is not unique unless $\stab_\phi (\mr{\xi})$ is trivial.
All other $Z\in\grpG$ with $\xi=\phi_Z(\mr{\xi})$ are of the form $Z=SX$ with $S\in \stab_\phi (\mr{\xi})$.
One verifies that for a given $X \in \grpG$ such that $\phi_X(\mr{\xi}) = \xi$ and any $S \in \stab_\phi (\mr{\xi})$ then
\begin{align}
\Lambda(\phi_{SX}(\mr{\xi}),v) & = \Ad_{(SX)^{-1}} \Lambda_{\mr{\xi}}(\psi_{(SX)^{-1}}(v)) \notag \\
%& = \Ad_{X^{-1}}\Ad_{S^{-1}} \Lambda_{\mr{\xi}}(\psi_{S^{-1}} \psi_{X^{-1}}(v)) \notag\\
& = \Ad_{X^{-1}} \Lambda_{\mr{\xi}}(\psi_S \psi_{S^{-1}} \psi_{X^{-1}}(v))
%& = \Ad_{X^{-1}} \Lambda_{\mr{\xi}}(\psi_{X^{-1}}(v)) \\
 = \Lambda(\phi_X(\mr{\xi}),v) \label{eq:lamndaprimeproof}
\end{align}
where the step to \eqref{eq:lamndaprimeproof} depends on \eqref{eq:stab_equivaraince}.

To see that $\Lambda$ is a lift compute
\begin{align}
\td \phi_\xi \Lambda(\xi,v) & = \td \phi_\xi \Ad_{X^{-1}} \Lambda_{\mr{\xi}}(\psi_{X^{-1}}(v)) \notag \\
& = \td \phi_X \td \phi_{\mr{\xi}} \Lambda_{\mr{\xi}}(\psi_{X^{-1}}(v)) \label{eq:using_homog_vel_prop} \\
& = \td \phi_X \calf(\mr{\xi},\psi_{X^{-1}}(v))
%\notag \\
%& =  \calf(\phi_X (\mr{\xi}),\psi_X \psi_{X^{-1}}(v)) \notag \\
 =  \calf(\xi,v) \notag
\end{align}
where line \eqref{eq:using_homog_vel_prop} follows from Proposition \ref{prop:homog_vel_transf}, and the last two steps follow from \eqref{eq:velocity_lift_xi0} and equivariance of $\calf$, respectively.

To see that $\Lambda(\xi,v)$ is equivariant, compute
\begin{align}
\Ad_{Z^{-1}} \Lambda(\xi,v) & = \Ad_{Z^{-1}} \Ad_{X^{-1}} \Lambda_{\mr{\xi}}(\psi_{X^{-1}}(v)) %\notag \\
%& = \Ad_{Z^{-1} X^{-1}} \Lambda_{\mr{\xi}}(\psi_{X^{-1} }\psi_{Z Z^{-1}} (v)) \notag\\
% = \Ad_{(XZ)^{-1}} \Lambda_{\mr{\xi}}(\psi_{ X^{-1}} \psi_{Z^{-1}}\psi_{Z}(v)) \notag \\
%& = \Ad_{(XZ)^{-1}} \Lambda_{\mr{\xi}}(\psi_{(Z^{-1} X^{-1})} \psi_{Z}(v)) \notag \\
 = \Ad_{(XZ)^{-1}} \Lambda_{\mr{\xi}}(\psi_{ (XZ)^{-1}} \psi_{Z }(v)) \notag \\
&  =  \Lambda(\phi_{XZ}(\mr{\xi}), \psi_{Z }(v)) %& =  \Lambda(\phi_{Z} \phi_{X}(\mr{\xi}), \psi_{Z }(v)) \notag \\
 =  \Lambda(\phi_{Z} (\xi), \psi_{Z }(v)) \label{eq:lambdaprimeproofequi}
\end{align}
where line \eqref{eq:lambdaprimeproofequi} follows from definition \eqref{eq:lambda_construct_one} with $X$ replaced by $XZ$, and the final line uses the definition of $\xi = \phi_{X}(\mr{\xi})$.
This completes the proof.
\end{proof}

Note that if $\stab_\phi(\xi)$  is trivial (the group action is free) then \eqref{eq:stab_equivaraince} holds implicitly.
Indeed, if  $\stab_\phi(\mr{\xi})$ is trivial then $\td \phi_{\mr{\xi}}$ is an isomorphism and
\begin{align}
\Lambda_{\mr{\xi}}(v) := \td \phi_{\mr{\xi}}^{-1} f(\mr{\xi},v)
\end{align}
is well defined and \eqref{eq:lambda_construct_one} defines an equivariant lift.
It is easily verified that in this case
\[
\Lambda(\xi,v) = \td \phi_\xi^{-1} f(\xi,v)
\]
for all $\xi\in\calM$ and $v\in\calV$.
In particular, for kinematic systems with Lie-group as state-space the lift function is just right translation of the system function back to the identity tangent space \cite{2020_Mahony_mtns}.
For a more general system where the state-space is not a Lie-group torsor, then existence of an equivariant lift is not obvious.

\begin{theorem}\label{th:equivariant_lift_existance}
Consider an equivariant linear system function $\calf : \calV \to \gothX(\calM)$ (Def.~\ref{def:equivariant_model}).
Then an equivariant lift $\Lambda : \calM \times \calV \to \gothg$ (Def.~\ref{def:equivarian_lift}) always exists.
\end{theorem}

\begin{proof}
Fix a point $\mr{\xi} \in \calM$.
We construct a function $\Lambda_{\mr{\xi}} : \calV \to \gothg$ by defining it successively on domains given by a chain of linear subspaces
$\calS_1 \subset \calS_2 \cdots \subseteq \calV$.

Choose an element $v_1 \in \calV$.
Since $\phi$ is transitive, then it is always possible to find an image  point  $\Lambda_{\mr{\xi}}(v_1) \in \gothg$ such that
%\[
$\td \phi_{\mr{\xi}} \Lambda_{\mr{\xi}} (v_1) = \calf(\mr{\xi},v_1)$.
%\]
By linearity then $\Lambda_{\mr{\xi}} (\alpha v_1) = \alpha \Lambda_{\mr{\xi}} (v_1)$ for all $\alpha\in\R$ and $\Lambda_{\mr{\xi}}$ is defined on the domain $\spn \{v_1\}$.

Define a $\psi_{\stab_\phi(\mr{\xi})}$ invariant subspace $\calS_1$ containing $v_1$ by
%\[
$\calS_1 = \spn \{ \psi_S(v_1) \;|\; S \in \stab_\phi(\mr{\xi}) \}$.
%\]
Define $\Lambda_{\mr{\xi}}$ restricted to the domain $\calS_1$ by the linear extension of the function defined by
%\[
$\Lambda_{\mr{\xi}} (\psi_S(v_1)) := \Ad_{S^{-1}} \Lambda_{\mr{\xi}} (v_1)$.
%\]
To see that this is well defined then one must verify that if $\psi_{S}(v_1) = \psi_T(v_1)$ for $S,T\in\stab_\phi(\mr{\xi})$ then $\Ad_{S^{-1}} \Lambda_{\mr{\xi}} (v_1) =
\Ad_{T^{-1}} \Lambda_{\mr{\xi}} (v_1)$.
However, in this case $\psi_{ST^{-1}}(v_1) = v_1$ and hence $ST^{-1}\in\stab_\phi(\mr{\xi})$ and
\begin{align}
\Ad_{T^{-1}} \Lambda_{\mr{\xi}} (v_1) & = \Ad_{T^{-1}} \Lambda_{\mr{\xi}} (\psi_{ST^{-1}}(v_1)) %\notag \\
 = \Ad_{T^{-1}} \Ad_{T S^{-1}} \Lambda_{\mr{\xi}} (v_1) %\notag \\
 = \Ad_{S^{-1}} \Lambda_{\mr{\xi}} (v_1) \label{eq:pf:Lambda_well_defined}
\end{align}
as required.

Note that
\begin{align}
\td \phi_{\mr{\xi}} \Lambda_{\mr{\xi}}(\psi_S(v_1))
& = \td \phi_{\mr{\xi}} \Ad_{S^{-1}} \Lambda_{\mr{\xi}}(v_1)
 = \left. \tD_X \right|_{\Id} \phi_{\mr{\xi}} \circ I_{S^{-1}} (X) [\Lambda_{\mr{\xi}}(v_1)] \notag\\
%& = \left. \tD_X \right|_{\Id} \phi( S^{-1} X S, {\mr{\xi}}) [\Lambda_{\mr{\xi}}(v_1)] \notag\\
& = \td \phi_{S} \td \phi_{\mr{\xi}} \Lambda_{\mr{\xi}}(v_1) %\notag\\
 = \td \phi_{S} \calf(\mr{\xi},v_1)% \notag\\
 = \calf ( \mr{\xi},\psi_{S}(v_1)) \label{eq:pf:Lambda_projection}
\end{align}
where the fact that $\phi_{S^{-1}}(\mr{\xi}) = \mr{\xi} = \phi_{S}(\mr{\xi})$ is used in the last two lines and the last line follows from equivariance of $\calf$.
By linearity this property extends to any element $w \in \calS_1$ and $\Lambda_{\mr{\xi}}$ has the lift property \eqref{eq:velocity_lift_xi0} on $\calS_1$. Property \eqref{eq:stab_equivaraince} holds by construction.

The proof proceeds by induction.
Let $\calS_j \subset \calV$ be a $\psi_{\stab_{\phi}(\mr{\xi})}$ invariant subspace on which $\Lambda_{\mr{\xi}}$ is defined.
If $\calS_j \not= \calV$ then choose $v_{j+1} \in \calV \setminus \calS_j$.
Find $\Lambda_{\mr{\xi}}(v_{j+1}) \in \gothg$ such that
%\[
$\td \phi_{\mr{\xi}} \Lambda_{\mr{\xi}} (v_{j+1}) = \calf(\mr{\xi},v_{j+1})$.
%\]
Noting that $\Lambda_{\mr{\xi}} (v_{j+1}) \in \gothg$ can be any element and does not need to be disjoint from the previous choices.
Define a subspace
%\[
$\calS'_{j+1} = \spn \{ \psi_S(v_{j+1}) \;|\; S \in \stab_\phi(\mr{\xi}) \}$
%\]
and note that $\calS'_{j+1}$ is $\psi_{\stab_\phi(\mr{\xi})}$  invariant.
Note furthermore that $\calS'_{j+1} \cap \calS_{j} = \{0\}$ since otherwise $\calS_{j}$ would not have been $\psi_{\stab_\phi(\mr{\xi})}$  invariant.
Extend $\Lambda_{\mr{\xi}}$ to $\calS'_{j+1}$ by
%\[
$\Lambda_{\mr{\xi}} (\psi_S(v_{j+1})) := \Ad_{S^{-1}} \Lambda_{\mr{\xi}} (v_{j+1})$ and linear extension.
%\]
Repeating the computation \eqref{eq:pf:Lambda_well_defined} verifies that $\Lambda_{\mr{\xi}}$ is well defined and repeating \eqref{eq:pf:Lambda_projection} verifies that the extended function $\Lambda_{\mr{\xi}}$ has the lift property.
Define
%\[
$\calS_{j+1} = \spn \{ \calS'_{j+1} + \calS_j \} \subset \calV$.
%\]
The lift $\Lambda_{\mr{\xi}}$ extends to $\calS_{j+1}$ through linearity.
It follows that $\calS_1 \subset \calS_2 \cdots$ is a strictly ascending chain of $\psi_{\stab_{\phi}(\mr{\xi})}$ invariant subspaces contained in $\calV$.
Let $\calS^1 = \cup_{j = 1}^\infty \calS_j$ be the limiting set and
note that $\Lambda_{\mr{\xi}}$ is well defined on $\calS^1$.

If $\calS^1 = \calV$ the proof terminates.
Otherwise, consider a point $v^2_1 \in \calV$ that is disjoint from $\calS^1$.
This new initial vector can be used to construct a new strictly ascending chain of subspaces $\calS^1 \subset \calS^2_1 \subset \calS^2_2 \cdots$ and lead eventually to a union of all these sets $\calS^2$.
The process can be continued \textit{ad-infinitum} to generate a collection of sets such that the union of any ascending chain is also an element of the collection.
Applying Zorn's lemma one concludes that there is a maximal element $\calS^\star$ of this collection.
However, if the maximal set $\calS^\star \not= \calV$ then one can always choose a new element $v^\star \in \calV \setminus \calS^\star $ and use the
same process to extend the collection beyond $\calS^\star$, a contradiction that $\calS^\star$ was a maximal element.
It follows that $\calS^\star = \calV$ and the resulting function $\Lambda_{\mr{\xi}} : \calV \to \gothg$ is defined on the whole input space $\calV$.

By construction, $\Lambda_{\mr{\xi}}$ satisfies \eqref{eq:velocity_lift_xi0} and \eqref{eq:stab_equivaraince} and the result follows from Lemma \ref{lem:equivariant_lambda}.
\end{proof}

The construction in Theorem \ref{th:equivariant_lift_existance} is finite if $\calV$ is finite-dimensional.
In practice, real kinematic models have the original finite-dimensional linear system function $f : \vecV \to \gothX(\calM)$ that is compatible with a complete velocity output (Def.~\ref{def:complete_output_g}) by assumption.
Clearly, this finite dimensional subspace will be the first point of call for choosing the germ vectors $v_i$ and a finite number of iterations will ensure that $\Lambda_{\mr{\xi}}$ will be fully defined on $\vecV$ and indeed on its extension by the stabilizer
\[
\vecV_{\stab_\phi(\mr{\xi})} = \{ w = \psi_S(v) \; |\; v \in \vecV, S \in \stab_\phi(\mr{\xi}) \}.
\]
It turns out that this development is sufficient for the implementation of practical observers and filters since the velocities measured always lie in $\vecV$.
Thus, although it is important to know that an equivariant lift exists, one only needs to compute the lift on a finite dimensional subspace of the input extension to implement a practical observer.

%%%%%%%%%%%%%%%%%%%%%%%%%%%%%%%%%%%%%%%%%%%%%%%%%%%%%%%%%%%%%%%%%%%%%%%%%%%%%%%
\section{Observer design}\label{sec:Observer_design}

The observer problem considered is to design a dynamical system, the observer, whose inputs are measurements taken from a kinematic system $\calB$ (Def.~\ref{def:kinematic_system}) and whose output $\hat{\xi}(t) \in \calM$ converges to the kinematic system's true state $\xi(t) \in \calB$ in a suitable manner.
This section provides solutions to steps \ref{item:observer_architecture}, \ref{item:error_dynamics} and \ref{item:observer_design} in the proposed observer design methodology \S\ref{sub:design_outline}.

%~~~~~~~~~~~~~~~~~~~~~~~~~~~~~~~~~~~~~~~~~~~~~~~~~~~~~~~~~~~~~~~~~~~~~~~~~~~~~~
\subsection{Observer Architecture}

\begin{figure}[h]
\begin{center}
\includegraphics[scale=0.4167]{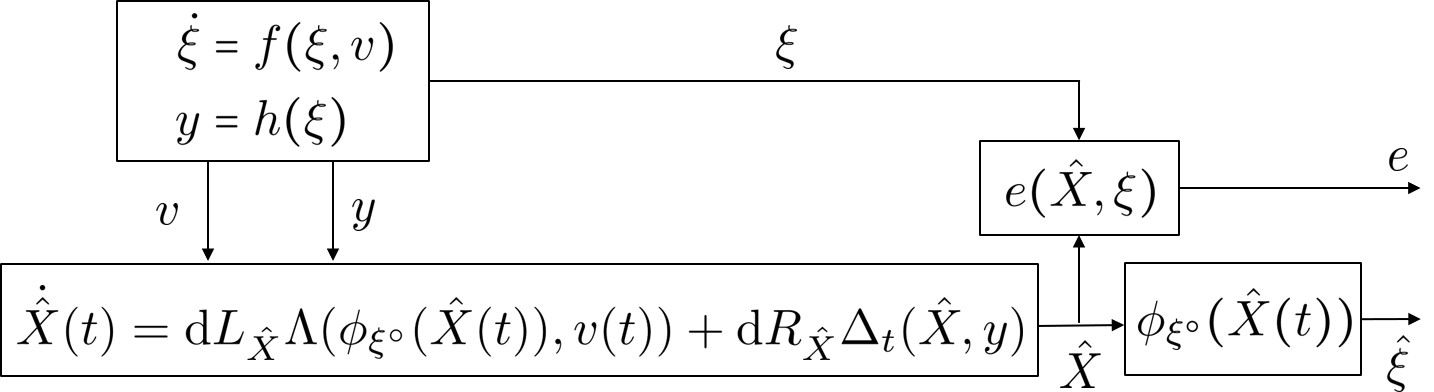}
% the scale is chosen to scale 24pt font to 10pt font.
\caption{
Proposed observer architecture.
The observer state $\hat{X} \in \grpG$ has internal model dynamics given by the lifted system \eqref{eq:lifted_system} and innovation function $\Delta_t : \grpG \times \calN \to \gothg$.
The observer state estimate function is $\phi_{\mr{\xi}} : \calG \to \calM$.
The error function $e : \grpG \times \calM \to \calM$ (Def.~\ref{def:error_function}) is a design signal that measures convergence of the observer state.
}
\label{fig:ObserverBD_v01}
\end{center}
\end{figure}

The proposed observer architecture is shown in Figure \ref{fig:ObserverBD_v01}.
In addition to the material developed earlier in the paper, configuration outputs, innovation functions, and error functions are required to formulate this architecture.
In this section, we discuss configuration outputs and innovation functions.
Error functions are discussed in Sections \S\ref{sub:error} and \S\ref{sub:invariant_error}.

\begin{definition}\label{def:config_output}
A \emph{configuration output} is a smooth map $h : \calM \rightarrow \calN$ into a smooth manifold $\calN$.
The configuration output is denoted by $y = h(\xi)$.
\end{definition}

The observer state is posed on the symmetry Lie group $\grpG$ and the lifted system is used as an internal model.
The proposed observer architecture is a classical internal model plus innovation design \eqref{eq:basic_observer}.

\begin{definition}\label{def:observer_architecture}
Consider a kinematic model (Def. \ref{def:kinematic_model}) with homogeneous state (Def.~\ref{def:symmetry_group}) and lifted system \eqref{eq:lifted_system}.
Choose an origin $\mr{\xi} \in \calM$.
The observer is defined to be
\begin{align}
\dot{\hat{X}}(t) = \td L_{\hat{X}} \Lambda(\phi_{\mr{\xi}} (\hat{X}(t)), v(t)) + \td R_{\hat{X}} \Delta_t(\hat{X},y) \quad \hat{X}(0) \in \grpG
\label{eq:basic_observer}
\end{align}
where $\hat{X} \in \grpG$ is the \emph{observer state} and $\Delta_t : \grpG \times \calN \rightarrow \gothg$ is a time varying \emph{innovation function} that remains to be defined.
The corresponding state estimate will be given by
\[
\hat{\xi}(t) := \phi_{\mr{\xi}} (\hat{X}(t)).
\]
\end{definition}

If $\hat{X}$ has the property that $\hat{\xi} = \phi_{\hat{X}}(\mr{\xi}) = \xi$ then choosing the innovation function $\Delta_t(\hat{X},y) = 0$ ensures that the trajectories of the observer state estimate track the trajectories of the true system state (Lemma~\ref{lem:project_system}).
In general, the innovation function\footnote{%
In classical observer design, the term innovation refers to the \emph{output residue} $(y - h(\hat{\xi}))$.
However, such a construction is only well defined when the output space is linear.
In a stochastic observer, the residue is mapped via a `Kalman' gain matrix $K_t$ to generate a correction term $K_t( y - h(\phi_{\hat{X}}(\mr{\xi})))$ that is added to the internal model in the observer ODE .
While a clean concept of the output residue requires additional structure, the \emph{innovation function} $\Delta_t : \grpG \times \calN \to \gothg$,
combining the Kalman gain and output residue concepts from classical filter design,
is a powerful formulation that aides a wide range of design approaches in equivariant systems theory for observer design.
}
is chosen to steer the trajectory $\hat{X}(t)$ to incrementally reduce an \emph{error} $e(\hat{X},\xi)$, from observer state to system state.
Choosing a suitable error is a crucial step in the design process that is discussed in Section \S\ref{sub:error} and \S\ref{sub:invariant_error}.
The overall architecture of the observer and system kinematics is shown in Figure \ref{fig:ObserverBD_v01}.

%~~~~~~~~~~~~~~~~~~~~~~~~~~~~~~~~~~~~~~~~~~~~~~~~~~~~~~~~~~~~~~~~~~~~~~~~~~~~~~
\subsection{Error functions}\label{sub:error}

In classical observer design, on linear spaces or in local coordinates, the ubiquitous choice of error is the difference $(\xi - \hat{\xi})$.
However, this error does not generalise well to the equivariant case.
Firstly, for an equivariant observer the observer state is $\hat{X} \in \grpG$ and not an element of the system state space $\calM$.
Even if the substitution $\hat{\xi} = \phi_{\mr{\xi}}(\hat{X})$ is made, the difference operation depends on linear or local coordinates for the state.
Finally, the naive error does not carry natural invariance properties and this destroys structure in the error dynamics (derived in Section~\S\ref{sub:observer_design}) that underly the effectiveness of the overall approach.
Choosing the appropriate error is a crucial design choice in equivariant systems theory for observer design.

To motivate the invariant error that we propose in Section \S\ref{sub:invariant_error}, we begin by a short discussion of the nature of error functions in general.
This digression is important since there are very few observer design papers that consider different observer and system state spaces, and as such the concepts behind defining an error between signals from different spaces with different dimensions are not consistent in the literature.

\begin{definition}\label{def:error_function}
Consider an observer with state space $\calG$ a smooth manifold, for a system with state space $\calM$ a smooth manifold, and with state estimate given by a smooth submersion $\Phi_0 : \calG \to \calM$.
A \emph{global error function} is a smooth function $e :\calG \times \calM \to \calM$ such that
\begin{itemize}
\item[i)]
the family of partial maps $e_{\hat{X}} :  \calM \to \calM$,
\begin{align}
e_{\hat{X}}(\xi) := e(\hat{X}, \xi)
\label{eq:error_represent}
\end{align}
are diffeomorphisms.

\item[ii)]
the family of partial maps $e_{\xi} :  \calG \to \calM$,
\begin{align}
e_{\xi}(\hat{X}) := e(\hat{X}, \xi)
\label{eq:error_submersive}
\end{align}
are submersions.
\end{itemize}
A global error function is termed \emph{consistent} with the observer if
there exists a constant $\mr{e} \in \calM$ for which
\begin{align}
e(\hat{X}, \Phi_0 (\hat{X}) ) = \mr{e}
\label{eq:error_consistent}
\end{align}
for all $\hat{X}\in\calG$.
\end{definition}

Given a consistent error function satisfying Def.~\ref{def:error_function} then if
$e(\hat{X},\xi)  = \mr{e}$ it follows that
\begin{align*}
e_{\hat{X}} (\xi) & = e(\hat{X},\xi)  = \mr{e} =  e(\hat{X}, \Phi_0 (\hat{X}) ) =  e_{\hat{X}} (\Phi_0 (\hat{X}))
\end{align*}
Since $e_{\hat{X}} : \calM \to \calM$ is a diffeomorphism, applying its inverse to both sides of the equation yields $\Phi_0(\hat{X}) = \xi$.
Thus, the observer objective will be to design the innovation to force $e(\hat{X},\xi) \to \mr{e}$.
Conditions \eqref{eq:error_represent} and \eqref{eq:error_submersive} ensure that the error is a global and effective measure of the observer state convergence.
It is easy to verify that the classical error $(\xi - \hat{\xi})$ is a consistent global error function that satisfies Def.~\ref{def:error_function} for systems on $\R^n$ with the identity state estimate map.

%~~~~~~~~~~~~~~~~~~~~~~~~~~~~~~~~~~~~~~~~~~~~~~~~~~~~~~~~~~~~~~~~~~~~~~~~~~~~~~
\subsection{Invariant Errors}\label{sub:invariant_error}

\begin{definition}\label{def:error_invariant}
Let $\phi : \grpG \times \calM \to \calM$ be a group action on a smooth state space manifold $\calM$ (Def.~\ref{def:symmetry_group}).
Let $\phi': \grpG \times \calG \to \calG$ be a transitive group action on a smooth observer state space manifold.
A global error function $e : \calG \times \calM \to \calM$ (Def.~\ref{def:error_function}) is said to be invariant if
\begin{align}\label{eq:error_gen_invariance}
e(\phi'_A(\hat{X}), \phi_A(\xi)) =  e(\hat{X}, \xi), \quad \text{ for all } A \in \grpG.
\end{align}
\end{definition}

Invariance in an error is highly desirable as it decouples error behaviour from the actual state of the system and measures only the relative estimate-to-state error.
An observer designed to decrease such an error is well conditioned over all possible initial conditions.
Without invariance, different initial conditions lead to different error behaviour for the same input signals and gain settings.
Tuning the gains for an observer derived from a non-invariant error becomes a case-by-case trajectory-by-trajectory problem requiring knowledgable engineers, specific targeted gains for different scenarios, and many corner cases for when the gains fail.
A key motivation for lifting onto the symmetry group and defining the internal model in group coordinates comes from the ability to use the group multiplication and group action to define an invariant observer error.

\begin{definition}\label{def:error}
Consider a kinematic model with homogeneous state (Def.~\ref{def:symmetry_group}) and lifted model
\eqref{eq:lifted_system}.
The \emph{state error} $e : \grpG \times \calM \rightarrow \calM$ is defined to be
\begin{align}\label{eq:e}
e = e(\hat{X},\xi) := \phi(\hat{X}^{-1},\xi).
\end{align}
The \emph{group error} $E: \grpG \times \grpG \rightarrow \grpG$ is defined to be
\begin{align}\label{eq:E}
E = E(\hat{X}, X) := X \hat{X}^{-1}.
\end{align}
\end{definition}

Recall that the lifted system \eqref{eq:lifted_system} is equivariant with respect to right translation (Lemma~\ref{lem:equivariance_lifted_system}) and note that the group error \eqref{eq:E} can also be written $E = R_{\hat{X}^{-1}} X$ in a similar form to \eqref{eq:e}.

\begin{lemma}\label{lem:invariant_error}
Consider a kinematic model with homogeneous state (Def.~\ref{def:symmetry_group}) and lifted model
\eqref{eq:lifted_system}.
The state error \eqref{eq:e} is an invariant (Def.~\ref{def:error_invariant}) and consistent global error (Def.~\ref{def:error_function}) with respect to observer state space group action $R : \grpG \times \grpG \to \grpG$ and system state space group action $\phi : \grpG \times \calM \to \calM$.
The group error \eqref{eq:E} is an invariant and consistent global error with respect to observer state space group action $R : \grpG \times \grpG \to \grpG$ and lifted system state space group action $R : \grpG \times \grpG \to \grpG$.
\end{lemma}

\begin{proof}
Since $\phi_{\hat{X}^{-1}} : \calM \to \calM$ and $R_{\hat{X}^{-1}}: \grpG \to \grpG$ are diffeomorphisms then the proposed errors satisfy \eqref{eq:error_represent}.
Since both actions are transitive they satisfy \eqref{eq:error_submersive}.

Choose $\hat{X}$ arbitrarily and evaluate
\begin{align}
e(\hat{X},\phi_{\mr{\xi}}(\hat{X}))
= \phi(\hat{X}^{-1},\phi_{\hat{X}}(\mr{\xi}))
= \phi(\hat{X} \hat{X}^{-1},\mr{\xi}) = \mr{\xi}. \label{eq:state_error_E}
\end{align}
Thus, \eqref{eq:e} satisfies \eqref{eq:error_consistent} with $\mr{e} = \mr{\xi}$ and $\Phi_0=\phi_{\mr{\xi}}$.
Similarly, $E(\hat{X},\hat{X}) = \hat{X} \hat{X}^{-1} = \Id$ and \eqref{eq:E} satisfies \eqref{eq:error_consistent} with $\mr{E} = \Id$ and $\Phi_0=\id$.

Invariance of the state error is seen by computing
\begin{align}
e(R_Z \hat{X}, \phi_Z(\xi)) & = e( \hat{X}Z , \phi_Z(\xi)) %\notag \\
 = \phi ( (\hat{X}Z)^{-1},\phi_Z(\xi))  %\notag  \\
%& = \phi (\hat{X}^{-1}, \phi((Z)^{-1}\phi_Z(\xi))) \notag  \\
 = \phi (\hat{X}^{-1},\xi) = e( \hat{X}, \xi).
\end{align}
Invariance of the group error is seen by computing
\[
E(R_A (\hat{X}), R_A(X)) = X A (\hat{X} A)^{-1} = X A A^{-1} \hat{X}^{-1} = X \hat{X}^{-1} = E(\hat{X}, X).
\]
\end{proof}

It is reasonable to question whether the two errors proposed in Definition \ref{def:error} are the only invariant error choices.
In particular, it would be desirable to find an error $e_\diamond : \calM \times \calM \to \calM$ defined directly on the state space that has an invariance property $e_\diamond(\hat{\xi}, \xi) = e_\diamond(\phi_A (\hat{\xi}), \phi_A(\xi))$.
Such an error would allow development of an observer directly on the state space $\calM$ and overcome many of the complexities associated with lifting the system onto the symmetry group and posing the observer on a higher dimensional state.
However, the following theorem demonstrates that no such error exists
unless $\calM$ and the symmetry group are diffeomorphic.

\begin{theorem}\label{th:invariant_error}
Consider a general observer for a system on $\calM$ expressed as a dynamical system on a general manifold $\calG$ with state estimate given by a function $\Phi_0 : \calG \to \calM$.
Let $e : \calG \times \calM \to \calM$ be a consistent global error function (Def.~\ref{def:error_function}).
Let $\phi : \grpG \times \calM \to \calM$ be an effective and transitive group action from a symmetry group $\grpG$ to $\calM$ (Def.~\ref{def:symmetry_group}).
Let $\phi': \grpG \times \calG \to \calG$ be a transitive group action on the observer state space.
Then the error $e$ is invariant \eqref{eq:error_gen_invariance}
only if $\calG \equiv \grpG$ are diffeomorphic.
\end{theorem}

\begin{proof}
Choose arbitrary $\hat{X} \in \calG$ and $\xi \in \calM$.
We assume that the error is invariant \eqref{eq:error_gen_invariance} and prove $\calG \equiv \grpG$.
Choose $Z \in \stab_{\phi'} (\hat{X})$ and note that
\begin{align*}
e(\hat{X},\xi) & = e(\phi'_Z (\hat{X}), \xi) %\\
 = e(\phi'_{Z^{-1}}  (\phi'_Z (\hat{X}) ), \phi_{Z^{-1}} (\xi)) %\\
 = e(\hat{X}, \phi_{Z^{-1}} (\xi)).
\end{align*}
Since the partial maps $e_{\hat{X}}$ are diffeomorphisms then  $\xi = \phi_{Z^{-1}} (\xi)$.
Since $\phi$ is effective and $\xi$ is arbitrary it follows that $Z = \Id$ and hence $\stab_{\phi'}(\hat{X}) = \{\Id\}$ is trivial.
The result follows since $\phi'$ is transitive on $\calG$ and $\phi'_{\hat{X}}$ defines a diffeomorphism from $\grpG$ to $\calG$.
\end{proof}

If the natural symmetry $\grpG$ of a homogeneous space $\calM$ has a non-trivial stabilizer, and as long as the action is effective, there is no invariant error $e : \calM \times \calM \to \calM$.
If the action is not effective initially then by factoring out by the normal subgroup associated with the non-effective part of the action one obtains a new quotient group on which the induced action on $\calM$ is well defined and effective and the theorem will apply to this case.
That is, the classical observer and filter construction, where the observer state $\hat{\xi} \in \calM$ is chosen \emph{a-priori} as a copy of the system state, can never be analysed using an invariant error.
The only exception is the special case where the symmetry group acts freely (with trivial stabilizer) and  hence $\grpG$ is diffeomorphic to $\calM$.
Theorem \ref{th:invariant_error} strongly motivates the observer architecture
given in Definition \ref{def:observer_architecture}.

%~~~~~~~~~~~~~~~~~~~~~~~~~~~~~~~~~~~~~~~~~~~~~~~~~~~~~~~~~~~~~~~~~~~~~~~~~~~~~~
\subsection{Observer Design}\label{sub:observer_design}

The proposed approach for observer design is to compute the error dynamics and design an observer to stabilize these dynamics.
The resulting innovation function is applied in the observer state equation and the state estimate is generated by the observer output equation.

\begin{lemma}
Consider a kinematic model with homogeneous state (Def.~\ref{def:symmetry_group}) and lifted model
\eqref{eq:lifted_system}.
Consider the observer \eqref{eq:basic_observer} for initial condition $\hat{X}(0) \in \grpG$.
The error dynamics (Def. \ref{def:error}) are given by
\begin{align}
\dot{e} & = \td \phi_{e} \Ad_{\hat{X}}\left(\Lambda(\xi(t),v(t)) - \Lambda(\phi_{\mr{\xi}}(\hat{X}(t)),v(t)) \right)
-  \td \phi_{e} \Delta_t(\hat{X},y)\label{eq:e_dot}
\end{align}
\end{lemma}

\begin{proof}
Consider the group error $E$ \eqref{eq:E} for the lifted kinematics of the system \eqref{eq:lifted_system} with an initial condition $X(0)$ such that $\phi(X(0),\mr{\xi}) = \xi(0)$.
Direct computation yields
\begin{align}
\ddt E(t) & = \ddt X \hat{X}^{-1}  = \dot{X} \hat{X}^{-1} - X \hat{X}^{-1}\dot{\hat{X}}\hat{X}^{-1} \notag \\
& = X \Lambda(\phi_{\mr{\xi}} (X), v) \hat{X}^{-1}
- X \hat{X}^{-1}
\left(  \hat{X} \Lambda(\phi_{\mr{\xi}} (\hat{X}), v)
+ \Delta_t(\hat{X},y) \hat{X} \right)\hat{X}^{-1} \notag  \\
& = X \hat{X}^{-1} \Ad_{\hat{X}} \Lambda(\phi_{\mr{\xi}} (X), v)
- X \hat{X}^{-1}
\Ad_{\hat{X}} \Lambda(\phi_{\mr{\xi}} (\hat{X}), v)
- X \hat{X}^{-1} \Delta_t(\hat{X},y)  \notag  \\
& = E \Ad_{\hat{X}} \left( \Lambda(\phi_{\mr{\xi}} (X), v)
- \Lambda(\phi_{\mr{\xi}} (\hat{X}), v) \right)
- E\Delta_t(\hat{X},y).  \label{eq:E_dot}
\end{align}

Note that
\begin{align*}
e(\hat{X},\xi) = \phi(\hat{X}^{-1},\xi) = \phi(\hat{X}^{-1},\phi(X,\mr{\xi})) =
\phi(X\hat{X}^{-1},\mr{\xi}) = \phi(E,\mr{\xi}).
\end{align*}
Thus, the error dynamics of $e$ are
\begin{align*}
\ddt e(t) & = \ddt \phi_{\mr{\xi}} (E)  = \td \phi_{\mr{\xi}} \dot{E} %\notag \\
 = \td \phi_{\mr{\xi}} \td L_{E}
\left[ \Ad_{\hat{X}} \left( \Lambda(\phi_{\mr{\xi}} (X), v)
- \Lambda(\phi_{\mr{\xi}} (\hat{X}), v) \right)
- \Delta_t(\hat{X},y) \right] \notag \\
& = \td \phi_{e}\Ad_{\hat{X}} \left( \Lambda(\phi_{\mr{\xi}} (X), v)
- \Lambda(\phi_{\mr{\xi}} (\hat{X}), v) \right)
- \td \phi_{e} \Delta_t(\hat{X},y).
\end{align*}
\end{proof}

The goal of the observer design is to find an innovation function $\Delta_t : \grpG \times \calN \rightarrow \gothg$ that asymptotically stabilises the error dynamics \eqref{eq:e_dot}, that is drives $e \rightarrow \mr{\xi}$.
On examination it is clear that this is still not a simple problem in its general form.
The error dynamics \eqref{eq:e_dot} can be written as a function of the error $e$ and known variables $\hat{X}$, $v$ and $y$
\[
\ddt e(t) = \td \phi_{e}\Ad_{\hat{X}} \left(\Lambda(\phi_{\hat{X}}(e),v) -  \Lambda(\phi_{\mr{\xi}} (\hat{X}), v) \right)
- \td \phi_{e} \Delta_t(\hat{X},y)
\]
by exploiting the definition of $e = \phi_{\hat{X}^{-1}}(\xi)$.
The explicit dependence of the error dynamics on the observer state $\hat{X}$ as well as the exogenous input $v(t)$, as well as the inherent non-linear dependence on the error $e$, makes this form of error dynamics highly coupled.

This is the point where the equivariant input extension plays a key role in the observer design.
Let $\calf : \calV \to \gothX(\calM)$ be an equivariant extension of the linear system function $f : \vecV \to \gothX(\calM)$.
Let $\Lambda : \calM \times \calV \to \gothg$ denote an equivariant lift for $\calf$ (Def.~\ref{def:equivarian_lift}).
It follows from \eqref{eq:equivariant_infinitesimal_lift} that
\begin{align}
\ddt e(t)
& = \td \phi_{e}\Ad_{\hat{X}} \left(\Lambda(\phi_{\hat{X}}(e),v) -  \Lambda(\phi_{\mr{\xi}} (\hat{X}), v) \right)
- \td \phi_{e} \Delta_t(\hat{X},y) \notag \\
& = \td \phi_{e} \left(\Lambda(\phi_{\hat{X}^{-1}} \phi_{\hat{X}}(e),\psi_{\hat{X}^{-1}}(v)) -  \Lambda(\phi_{\hat{X}^{-1}} \phi_{\mr{\xi}} (\hat{X}), \psi_{\hat{X}^{-1}}(v)) \right)
- \td \phi_{e} \Delta_t(\hat{X},y) \notag \\
& = \td \phi_{e} \left(\Lambda(e,\psi_{\hat{X}^{-1}}(v)) -  \Lambda(\mr{\xi}, \psi_{\hat{X}^{-1}}(v)) \right)
- \td \phi_{e} \Delta_t(\hat{X},y). \label{eq:equivariant_e_dot}
\end{align}
The formal structure of this design problem is now highly tractable.
The error dynamics depend only on the error $e$ along with an exogenous input $w(t) = \psi_{\hat{X}^{-1}}(v(t))$.
Although this input depends in turn on the observer state it can be viewed as an external input from the point of view of analysing the error dynamics.
Note that the input $w(t) = \psi_{\hat{X}^{-1}}(v(t))$ depends on the equivariant input extension and this construction is not possible without the theory developed in \S\ref{sec:equivariant_representation} in general.
An example of this process applied to a real world system is available in our recent work \cite{2020_Vangoor_cdc}.

In summary, the observer design problem for a kinematic model (Def.~\ref{def:kinematic_model}) with homogeneous state (Def.~\ref{def:symmetry_group}) can be tackled by an observer with architecture given by Definition~ \ref{def:observer_architecture}.
The innovation function $\Delta_t$ is designed to stabilise the invariant error \eqref{eq:e} dynamics
\begin{align}
\dot{e} = \td \phi_e \left( \Lambda(e,\psi_{\hat{X}^{-1}}(v)) - \Lambda(\mr{\xi},\psi_{\hat{X}^{-1}}(v)) \right) - \td \phi_e \Delta_t( \hat{X},h(\phi_{\hat{X}}(e))
\label{eq:equivariant_error_dynamics}
\end{align}
where $y = h(\phi_{\hat{X}}(e))$.
The actual choice of innovation function will of course depend on the specific system considered and the preferences of the design engineer.
There are several good methodologies available in the literature for this final step in the design.
The equivariant filter proposed by the authors
\cite{2020_Mahony_mtns,2020_Vangoor_cdc} uses a linearisation of \eqref{eq:equivariant_error_dynamics} along with Kalman-Bucy design principles.
The IEKF \cite{2008_Bonnabel_TAC,2017_Barrau_tac} of Bonnabel \etal is developed for systems on the Lie-group directly and uses an extended Kalman filter framework.
The recent paper \cite{2017_Barrau_tac} characterises the class of group affine systems for which the IEKF is applicable, and on these systems, the error dynamics \eqref{eq:equivariant_error_dynamics} expressed on the Lie-group are of a form where linearisation in the error coordinates is equivalent to linearisation along trajectories of the observer state (the extended Kalman filter perspective).
Constructive Ricatti observers have also been developed for a number of systems \cite{HamSam2017}.
It is often possible to use the structure of specific examples to build tailored nonlinear observer designs using constructive nonlinear design principles \cite{hamel2006,bonnabel2006,2008_Mahony_tac,2011_Hamel,2011_Madgwick,2011_Hua_SE(3),trumpf2012,2012_Grip,2013_Bras,2014_Sanyal,2015_Hua,2017_Berkane_Tayebi_Tac,2019_Hua,2020_Hua}.
This approach leads to observers with large (often almost global) basins of attraction and very high robustness factors but depends on case-by-case design.

%%%%%%%%%%%%%%%%%%%%%%%%%%%%%%%%%%%%%%%%%%%%%%%%%%%%%%%%%%%%%%%%%%%%%%%%%%%%%%%%
\section{Conclusion}\label{sec:conclusion}

In this paper, we have developed the foundational theory of equivariant systems that underlies the design of observers for systems with homogeneous state.
The effort taken to provide a strong systems theoretic development of the system model provides a foundation for future developments of the theory.
The core contribution of the paper lies in the theory of system lifts and equivariant input extensions that allows an equivariant observer design for any kinematic system on a homogeneous space.
We also provide key results on the existence of invariant errors and show how these integrate into observer design for systems with homogeneous state.
In particular, we motivate the choice to pose the observer state on the symmetry group rather than on the system state manifold.
The proposed approach is summarised in \S\ref{sub:design_outline}.

%%%%%%%%%%%%%%%%%%%%%%%%%%%%%%%%%%%%%%%%%%%%%%%%%%%%%%%%%%%%%%%%%%%%%%%%%%%%%%%%
\section*{Acknowledgments}
This research was supported by the Australian Research Council
through the ``Australian Centre of Excellence for Robotic Vision'' CE140100016 and the CNRS trough the the IRP-ARS (Advanced Autonomy for Robotic Systems).
%%%%%%%%%%%%%%%%%%%%%%%%%%%%%%%%%%%%%%%%%%%%%%%%%%%%%%%%%%%%%%%%%%%%%%%%%%%%%%%%
%===============================================================================

%===============================================================================
%% bibliography
%% Use the \bibliographystyle{alpha} to compile the bibligraphy in the main directory.
%\bibliographystyle{alpha}
%\bibliography{References}
%% Fetch the .bbl file from the directory above.
%% copy the .bbl file directly in here 

%===============================================================================

%%%%%%%%%%%%%%%%%%%%%%%%%%%%%%%%%%%%%%%%%%%%%%%%%%%%%%%%%%%%%%%%%%%%%%%%%%%
\end{document}